\documentclass[10pt,journal,compsoc]{IEEEtran}

\usepackage{graphicx}
\usepackage{subfigure}
 \usepackage{algorithm}
\usepackage{algorithmic}
\usepackage{booktabs} 
\usepackage{amsmath}
 \usepackage{multirow}
 \usepackage{rotating}
 \usepackage{amssymb}
 \usepackage{bbding}
 \usepackage{amsthm}
 \usepackage{bm}
 \usepackage{epstopdf}
 \usepackage{enumerate}
 \usepackage{array}
\usepackage{hyperref}
\usepackage{epstopdf}

\ifCLASSOPTIONcompsoc
  \usepackage[nocompress]{cite}
\else
  \usepackage{cite}
\fi

\newcommand{\ie}{\emph{i.e., }}

\ifCLASSINFOpdf
\else
\fi

\begin{document}

\title{SamWalker++: recommendation with informative sampling strategy}

\author{Can Wang,
        Jiawei Chen, Sheng Zhou, Qihao Shi, Yan Feng and Chun Chen\IEEEcompsocitemizethanks{\IEEEcompsocthanksitem Can Wang,
        Jiawei Chen, Sheng Zhou, Qihao Shi, Yan Feng and Chun Chen are with the School
of Computer Science and Technology, Zhejiang University, HangZhou,
China \\
Jiawei Chen is corresponding author\\
E-mail: \{wcan,sleepyhunt,zhousheng\_zju,shiqihao321,fengyan,chenc\}@zju.\\edu.cn
}}


\markboth{Transactions on Knowledge and Data Engineering}%
{Shell \MakeLowercase{\textit{et al.}}: Bare Demo of IEEEtran.cls for Computer Society Journals}

\IEEEtitleabstractindextext{%
\begin{abstract}
Recommendation from  \textit{implicit feedback} is a highly challenging task due to the lack of reliable negative feedback data. Existing methods address this challenge by treating all the un-observed data as negative (dislike) but downweight the confidence of these data. However, this treatment causes two problems: (1) Confidence weights of the unobserved data are usually assigned manually, which lack flexibility and may create empirical bias on evaluating user's preference. (2) To handle massive volume of the unobserved feedback data, most of the existing methods rely on stochastic inference and data sampling strategies. However, since a user is only aware of a very small fraction of items in a large dataset, it is difficult for existing samplers to select \textit{informative} training instances in which the user really dislikes the item rather than does not know it.

To address the above two problems, we propose two novel recommendation methods SamWalker and SamWalker++ that support both adaptive confidence assignment and efficient model learning. SamWalker models data confidence with a social network-aware function, which can adaptively specify different weights to different data according to users'  \textit{social contexts}. However, the social network information may not be available in many recommender systems, which hinders application of SamWalker. Thus, we further propose SamWalker++, which does not require any side information and models data confidence with a constructed pseudo-social network. In the pseudo-social network, similar users are connected with specific item nodes or community nodes. This way, the inference of one's data confidence can benefit from the knowledge from other similar users. We also develop fast random-walk-based sampling strategies for our SamWalker and SamWalker++ to adaptively draw informative training instances, which can speed up gradient estimation and reduce sampling variance. Extensive experiments on five real-world datasets demonstrate the superiority of the proposed SamWalker and SamWalker++.
\end{abstract}

}

\maketitle

\IEEEdisplaynontitleabstractindextext

\IEEEpeerreviewmaketitle

\IEEEraisesectionheading{\section{Introduction}\label{sec:introduction}}
With the exponential growth of information on electronic commerce websites, Collaborative Filtering (CF) as a prevalent approach in recommender systems are drawing more and more attention from both academia and industry \cite{ricci2015recommender,jannach2010recommender}. There are two types of feedback data in Collaborative Filtering systems. The first is called \emph{explicit feedback}, where the numerical ratings directly reflecting users' preference are provided. The other is \emph{implicit feedback}, which is a natural byproduct of users' behavior such as consumption, viewing or clicking. Since implicit feedback are more easily available, recent research attention is increasingly shifted from explicit feedback to implicit feedback. However, learning a recommender system from implicit feedback is more challenging due to the lack of reliable negative data. Only the positive feedback are observed, while the negative feedback are mixed with missing values in unobserved data. In other words, items interacted by the user reflect that the user favors the items, while non-interacted items does not necessarily mean the user dislikes the items. In most cases, users may just not know the items that they have not interacted.

A conventional strategy to address the problem is treating all the un-observed data as negative (dislike) but downweight the confidence of these data. However, this treatment poses two key research problems for implicit recommendation:

\textbf{(P1) How to assign  appropriate confidence weights for data?} Data confidence weights, which controls the contribution of the data on learning a recommendation model, usually significantly affect the model's accuracy. However, assigning appropriate confidence weights is challenging, as the real data confidence may change for various user-item combinations. Some unobserved data can be attributed to user's preference while others are the results of users' limited scopes. Most of existing methods rely on manual assignment of confidence weights to the data. Choosing confidence weights usually require human rich experience or large computational resource for grid search. Furthermore, it is unrealistic for researchers to manually set flexible and diverse weights for millions of data. Coarse-grained manual confidence weights will create empirical bias on estimating user's preference.

\textbf{(P2) How to efficiently learn a recommendation model from the large-scale implicit feedback data?} The large-scale unobserved data incur inefficiency problem. It is computationally impractical to traverse over the whole data set to obtain the gradients. To address this problem,  two types of strategies have been adopted in previous works. The first is batch-based gradient descent with memorization, such as ALS\cite{hu2008collaborative}, eAlS\cite{he2016fast}, FAWMF\cite{chen2020fast}. However, this kind of methods are only suitable for the specific models with K-separable property and L2 loss function, which will sacrifice models capacity and lead to sub-optimal performance. In fact, the models with more flexible deep structure, confidence weights and loss function have been validated achieving better performance \cite{he2017neural,chen2018modeling}. The other type is employing stochastic gradient descent solvers with data sampling strategies. However, in real-world applications, users typically are only aware of a relatively small fraction of the potential items \cite{lichman2018prediction}. In such cases, existing samplers usually select uninformative data with low confidence weights, in which the user just does not know the item rather than dislikes it. This will affect convergence and recommendation performance of the model.

\begin{figure}[t!]
\centering
\includegraphics[width=0.44\textwidth]{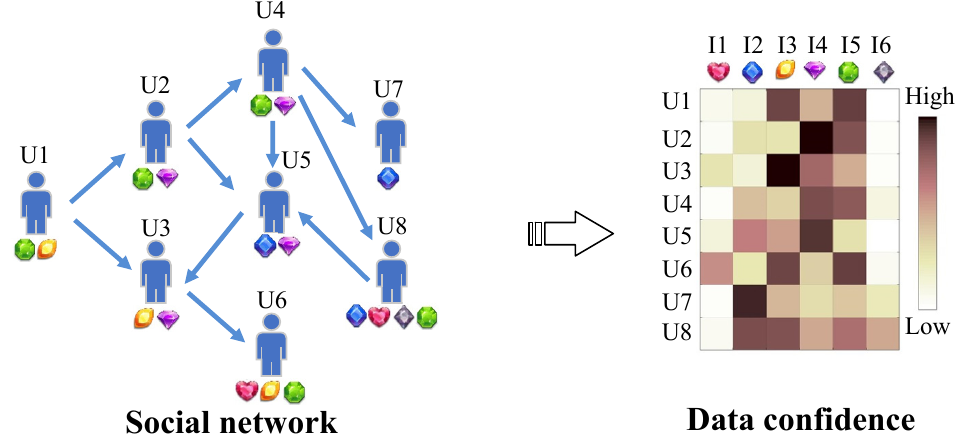}
 \caption{ SamWalker estimates data confidence based on user's social relations. The left part of the figure illustrates a social network including implicit feedback expressed by users. The consumed items (i.e. the items with positive feedback) are shown below the users. The right part shows our inferred data confidence.}
\label{samshiyi}
\end{figure}

\begin{figure}[t!]
\centering
\includegraphics[width=0.44\textwidth]{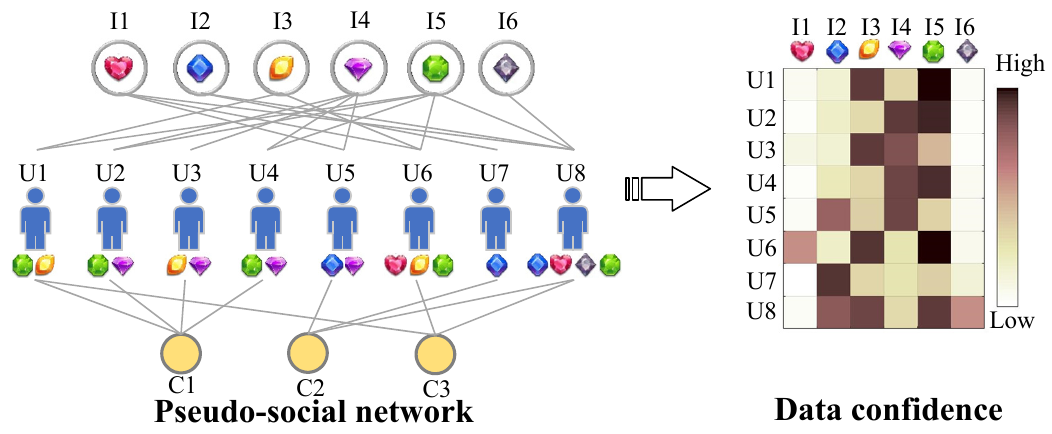}
 \caption{ SamWalker++ estimates data confidence based on the constructed pseudo-social network. The left part of the figure illustrates a pseudo-social network, where users are connected with additional item nodes ($I$1-$I$6) or community nodes($C$1-$C$3). We give links for the user-item pairs with positive feedback, and the user-community pairs if the user belongs to the community.}
\label{jiashiyi}
\end{figure}

To deal with these problems, we propose two novel recommendation methods SamWalker and SamWalker++ to simultaneously learn the personalized data confidence and draw informative training instances. We first present SamWalker which leverages social network information to address the problems. With the development of online social websites \cite{nie2016learning}, social relations have become a major information resource when users select items to consume \cite{Pan2017}. As we can see from Figure \ref{fg:dataana}, the social connected users exhibit more similarity in their consumptions on three typical recommendation datasets. Users usually get item information from social friends \cite{chen2013information} and their \textit{exposure} to items (i.e. whether a user knows the items) will inevitable be dominated by their \textit{social contexts} (i.e. whether their direct or indirect social neighbors have consumed the items).  Thus, users' social relations and social contexts reflect how users are exposed to the items and suggest the confidence of the data. It is consistent with our intuitions. Note that there exists two reasons for negative feedback: unknown or dislike. The more popular an item is among the user's social neighbors (e.g. the purple gem $I$4 comparing with the green gem $I$5 for user $U$1 in Figure \ref{samshiyi}), the more likely it will be that the user knows the item and his feedback is attributed to his preference. Correspondingly, the data will be more reliable in deriving user's preference. To capture this insight, as illustrated in Figure \ref{samshiyi}, SamWalker simulates item information propagation along the social network and models individual confidence weights as a social context-aware function. By iteratively learning transformation function and user's preference based on EXMF framework (exposure-based matrix factorization \cite{liang2016modeling}), SamWalker can adaptively specify different weights to different data based on user's social contexts.

A key limitation of SamWalker is that it requires the presence of social network information, which may not be available in many applications. To deal with this problem, we further propose SamWalker++, which only uses implicit feedback data and does not require any side information. The rationale of SamWalker++ is the ``wisdom of the crowds'', i.e. a user's behavior reflects not only his exposure but also the knowledge of other similar users. SamWalker++ constructs a pseudo-social network to replace social network, where users are connected with specific additional nodes, so that the knowledge of one's exposure can be transferred to other similar users. As shown in Figure \ref{samshiyi}, here we explicitly introduce two kinds of nodes to capture two kinds of similarities. On the one hand, note that a positive feedback signifies that the user knows the item. The users who have interacted with common items may have similar exposure. Thus, SamWalker++ leverages items as bridges so that the inference of the user's exposure can benefit from the rich information from the similar users with co-purchased items. On the other hand, recent social literatures \cite{palla2005uncovering,zhou2011understanding} suggest that users are clustered into some content-sharing communities. Item information will be spread in the community and the community members tend to share similar exposure. Motivated by this point, SamWalker++ deduces the latent communities for users and leverages communities as medium to exploit the knowledge of other community members. With the pseudo-social network, SamWalker++ devises a novel exposure model on the network and specifies data confidence weights with a network-aware function, which naturally encodes rich correlation information between users into the data confidence and potentially boosts recommendation performance.

\begin{figure}[t!]
\centering
\includegraphics[width=0.48\textwidth]{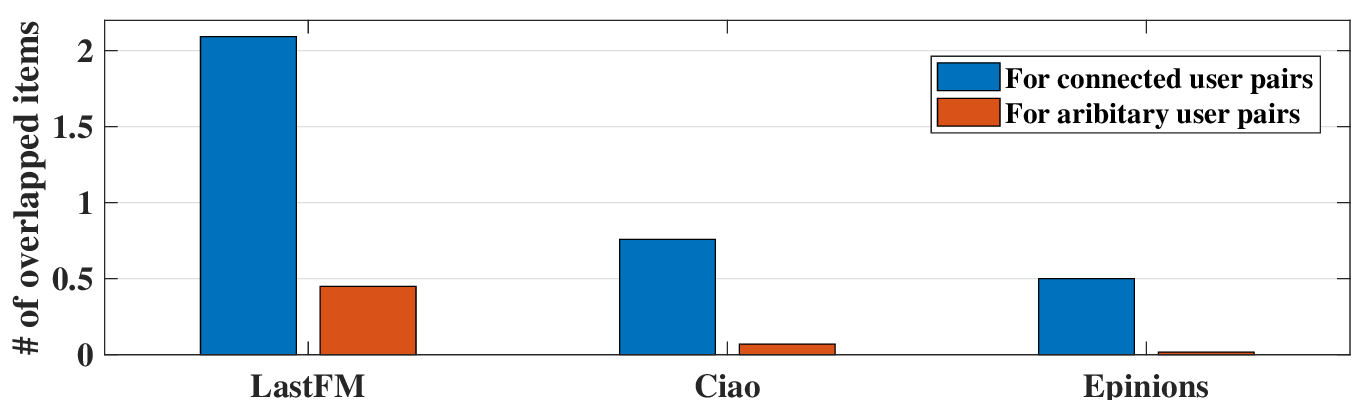}
 \caption{The average number of the overlapped items (the items that have been consumed by both users) for social connected user pairs and for arbitrary user pairs on three typical recommendation datasets.}
\label{fg:dataana}
\end{figure}

Let us use an example to illustrate why constructing pseudo-social network can help us infer the data confidence. Figure \ref{jiashiyi} illustrates an example of pseudo-social network. The inference of the target user's exposure can refer to the behavior of graph neighbors on the network. For example, we can deduce that user $U$3 may know the item $I$5 (green gem), because his similar users $U$1, $U$3, $U$4 who also belong to community $C$1 have interacted with $I$5. Correspondingly, his feedback data ($U$3-$I$5) can be attributed more to the preference and therefore has relatively large confidence weight. Another example can be seen from the higher confidence weights of the user-item pair ($U$1-$I$4) over the pair ($U$1-$I$2), as the similar users $U$2, $U$4 who share co-purchased items with $U$1 have consumed the item $I$4. Thus, the pseudo-social network is an effective and efficient tool to transfer the knowledge from the connected or even higher-order connected similar users to the target user, which improves the accuracy of the learned data confidence.

Due to the large number of unobserved data, developing an efficient informative sampling strategy is crucial. It is computationally infeasible to estimate and rank the current learned confidence weights for every data to select informative data. Instead, we propose efficient sampling strategies based on the random walk along the network for our SamWalker and SamWalker++. Intuitively, the more and closer neighbors have consumed the item, the more likely the user will know the item, in which case the feedback can be more confidently attributed to the user's preference. Consequently, we conduct personalized random walk for each user to explore his local (pseudo-) social network contexts and pick out items consumed by those similar users. Theoretical analysis proves that the distribution of the proposed sampling strategy is proportional to the data confidence while the sampling complexity is linear to the number of sampled instances, which heavily reduces the sampling variance and speed up gradient estimation.

Although this paper is extension of our previous work \cite{chen2019samwalker}, in which we present a social recommendation model SamWalker that adaptively learns the data confidence based on users' social context. In this article, we further deliver the following contributions:
\begin{itemize}
\item As the social network information may not be available in many recommender systems, we propose to construct pseudo-social network to replace social network, where similar users are connected with specific item nodes and community nodes.
\item We propose a novel recommendation model SamWalker++ on the pseudo-social network, which does not require any side information and can adaptively deduce the data confidence with the knowledge from other similar data.
\item We develop an efficient random walk-based sampling strategy along the pseudo-social network to draw informative training instances for SamWalker++, which can both reduce sampling variance and speed up gradient estimation.
\item Extensive experiments on five well-known benchmark datasets demonstrate that SamWalker++ outperforms a range of state-of-the-art methods and show the superiority of the proposed sampling strategy.
\end{itemize}

The rest of this paper is organized as follows. We briefly review related works in section 2. We give the problem definition and background in section 3. The SamWalker model is introduced in section 4. We further present our novel non-social method SamWalker++ in section 5. The informative sampling strategy and learning algorithm are presented in section 6. The experimental results are presented in section 7. Finally, we conclude the paper and present some directions for future work in section 8.

\section{Related work}
In this section, we review the most related works from the following four perspectives.

\textbf{Data confidence in implicit recommendation.} As the unobserved data are unreliable, learning a recommendation model from implicit feedback requires assigning confidence weights for the data. Most of the existing methods assign confidence weights manually. For example, the classic weighted matrix factorization (WMF) \cite{hu2008collaborative} and many neural-based collaborative methods (e.g. CDAE\cite{wu2016collaborative}, NCF\cite{he2017neural}, LightGCN\cite{he2020lightgcn}) used a simple heuristic where all negative feedback data are equally downweighted vis-a-vis the positive feedback data; \cite{he2016fast} and \cite{yu2017selection} assign the confidence weights based on item popularity. The readers can refer to the survey \cite{chen2020bias} for more information.

More recently, a new probabilistic model EXMF\cite{liang2016modeling} was proposed to incorporate user's exposure to items into the CF methods. When inferring user's preference, EXMF can translate user's exposure as data confidence. However, as analysed in section \ref{se:exmf}, this method suffers from over-fitting and low efficiency. \cite{chen2020fast} further proposes to model data confidence with a community-based neural network in their FAWMF. However, FAWMF is designed for fast memorization-based learning so that the capacity of the network is constrained. It can be seen from the low rank of the confidence matrix. The confidence weights are modeled from a global perspective and usually lack personality. Our experiments show that FAWMF performs poorly on the sparse dataset.

\textbf{Efficient Recommendation.} For efficient recommendation, two strategies have been proposed in previous works. The first is employing stochastic gradient descent and sampling strategy to accelerate learning. The most popular sampling strategy is to draw un-observed feedback data uniformly, which is applied in many recommendation models, including classic matrix factorization \cite{hu2008collaborative}, pair-wised models (e.g., Bpr  \cite{rendle2009bpr}) and sophisticated neural network-based methods (e.g., NCF \cite{he2017neural}, lightGCN\cite{he2020lightgcn}). However, uniform sampler will cause high variance and poor convergency.

 Some other sampling strategies are proposed to improve convergence from different perspectives: \cite{yu2017selection} and \cite{hernandez2014stochastic} propose item popularity-based and item-user co-bias sampling strategy to reduce sampling variance; \cite{chen2017sampling} presents several sampling strategies to balance backward computation of the item-dependent neural network and the user-item interaction function; \cite{yu2018walkranker} attempts draw positive data based on random walk along user-item bipartite graph. Our sampling strategies differs from \cite{yu2018walkranker} in that we pay more attention to sample informative negative data; Some work \cite{rendle2014improving, zhang2013optimizing, DBLP:conf/www/WangX000C20} also present subtle dynamic sampling strategies to over-sample the ``difficult'' negative examples in which the prediction is much different from the ground-truth.  Although effective, sampling ``difficult'' data for advanced preference model still suffers from low efficiency.  Also, the ``difficult" instances do not suggest that the user really dislike the item. The stochastic gradient estimator is biased and the natural noise in user-item feedback data may be amplified \cite{li2018adaerror}; To capture real negative data, \cite{ding2018improved,ding2019reinforced} propose to leverage exposure data, which however is not available in many situations.

 Another strategy for efficient recommendation is memorization strategy. When learning a recommendation model from implicit feedback, Some works \cite{hu2008collaborative,he2016fast,bayer2017generic,chen2020fast} propose to memorize some important intermediate variables so that the massive repeated computation can be avoided. However, these strategies are just suitable for the model with K-separable property and L2 loss function. In fact, more flexible loss functions(e.g. cross-entropy loss), confidence weights and neural network structure (e.g. NCF, LightGCN) have been validated achieving better performance.

\textbf{Social recommendation.} Social information has been utilized to improve recommendation performance in recent works. These methods mainly assume that connected users will share similar preference \cite{golbeck2009trust}. Sorec \cite{ma2008sorec}, TrustMF \cite{yang2013social}, PSLF \cite{shen2012learning}, jointly factorize rating matrix and trust (social) matrix by sharing a common latent user space; In \cite{chaney2015probabilistic,xiao2017learning}, users' feedback is considered as synthetic results of their preference and social influence;  \cite{jamali2010matrix,wang2017learning} utilize a social regularization term to constrain user's latent preference close to his trusted friends; \cite{wang2016social,zhao2014leveraging} extend pair-wise BPR framework by further assuming that for all items with negative feedback, a user would prefer the items consumed by their friends over the rest; \cite{jiawei2018social} also leverage social information to handle non-random missing data.

Also, there are two recent works claim that comparing with users' preference, users' exposure is more influenced by their social friends (neighbors). Thus, \cite{DBLP:conf/aaai/WangZYZ18} and \cite{chen2018modeling} integrate social influence on user's exposure into the generative process of EXMF model. However, these two methods need to infer large number of  parameters of user's exposure, which will suffer from overfitting and inefficiency problems.

\textbf{Random walk in recommendation.} Random walk strategy has been widely applied in recommendation. \cite{jamali2009trustwalker} performs random walk along the social network to search relevant users who have similar preference with the target user for better rating prediction;  \cite{christoffel2015blockbusters}  exploits random walk to obtain diverse recommendation; \cite{vahedian2017weighted} further extends \cite{christoffel2015blockbusters}  in heterogenous information network to generate valuable meta-paths; \cite{yu2018walkranker} employs random walk to find more positive instances. We remark that these works adopt static (uniform or pre-defined) transfer probability in their random walk strategy. Besides, these random walks are not designed for sampling informative training instances.

\section{Preliminaries}
\label{se:exmf}
In this section, we first give the problem definition of recommendation with implicit feedback. Then, we introduce exposure-based matrix factorization (EXMF) \cite{liang2016modeling} framework from a variational perspective to provide usual insight into the relation between user's exposure and data confidence.
\subsection{Problem definition}
Suppose we have a recommender system with user set $U$ (including $n$ users) and item set $I$ (including $m$ items). The implicit feedback data is represented as $n \times m$ matrix $X$ with entries $x_{ui}$ denoting whether or not the user $u$ has interacted with (e.g. consume\footnote{Throughout the paper, we will use the term consume to denote any kind of implicit interaction, unless otherwise stated.}) the item $i$. Social information represented as $n \times n$ matrix $T$, with $T_{ij}$ indicating connection between user $i$ and $j$. Also, $\mathcal T_u$ denotes the set of the connected social friends (direct neighbors) of the user $u$. The task of a recommender system can be stated as follow: recommending items for each user that are most likely to be consumed by him. The notations of this work are summarized in Table \ref{notation}.

\begin{table*}[t!]
\footnotesize
    \centering
    \caption{Notations and Definitions.}
    \begin{tabular}{clccc}
    \hline
    Notation           &   Annotation     \\
    \hline
    $U $& User set  \\
    $I$& Item set \\
    $n$ & The number of users in the system \\
    $m$ & The number of items in the system \\
    $X$ & User-item feedback matrix \\
    $\mathcal T_u$ & The set of connected friends of the user $u$ \\
    $a_{ui}$ & The variable denoting whether the user $u$ knows the item $i$ \\
    $\gamma_{ui}$ &The variational parameter of the variational posterior $q(a_{ui}|x_{ui})$, \ie $q(a_{ui}|x_{ui})=Bernoulli(\gamma_{ui})$ \\
    $Y$ &The matrix consisting of $\gamma_{ui}$ \\
    $m_{uj}^{u \leftarrow i \leftarrow u}$ &The information flowing from other similar users to the target user $u$ along the co-purchased items \\
    $m_{uj}^{u \leftarrow c \leftarrow u}$ &The information flowing from other similar users to the target user $u$ along the community nodes \\
    $\varphi_{ui}$ & The edge weight for the social relation \\
    $\varphi_{ui}^{u \leftarrow i}, \varphi_{iu}^{i \leftarrow u}, \varphi_{uc}^{u \leftarrow c}, \varphi_{cu}^{c \leftarrow u}$ & the weights for the edges $u \leftarrow i$, $i \leftarrow u$, $u \leftarrow c$ or $c \leftarrow u$ in the pseudo-network respectively \\
    $\Phi_{ui}, \Phi_{ui}^{u \leftarrow i}, \Phi_{iu}^{i \leftarrow u}, \Phi_{uc}^{u \leftarrow c}, \Phi_{cu}^{c \leftarrow u}$ & the matrixes that consist of $\varphi_{ui}, \varphi_{ui}^{u \leftarrow i}, \varphi_{iu}^{i \leftarrow u}, \varphi_{uc}^{u \leftarrow c}$ or $\varphi_{cu}^{c \leftarrow u}$ respectively   \\
    \hline
    \end{tabular}%
    \label{notation}
\end{table*}

\subsection{Exposure-based matrix factorization (EXMF)}

 EXMF \cite{liang2016modeling} directly incorporates user's exposure into collaborative filtering. This is achieved by first generating the latent variable $a_{ui}$, which indicates whether user $u$ has been exposed to item $i$. Then, EXMF models user's consumption $x_{ui}$ based on $a_{ui}$ as follow:
\begin{align}
   {a_{ui}}&\sim Bernoulli({\eta _{ui}}) \hfill \\
  ({x_{ui}}|{a_{ui}} = 1)&\sim Bernoulli(\sigma (p_u^{\top}{q_i})) \hfill \\
  ({x_{ui}}|{a_{ui}} = 0)&\sim {\delta _0}\approx Bernoulli(\varepsilon)\hfill
\end{align}
where $\delta _0$ denotes $p(x_{ui}=0|a_{ui}=0)=1$; ${\eta _{ui}}$ is the prior probability of exposure. Here we relax function $\delta _0$ as $Bernoulli(\varepsilon)$ to make model more robust, where $\varepsilon$ is a small constant (e.g. $\varepsilon$=0.001). When $a_{ui}=0$, we have $x_{ui}\approx 0$, since when the user does not know the item he can not consume it. When $a_{ui}=1$, when the user has learned the item, he will decide whether or not to consume the item based on his preference. $x_{ui}$ can be generated with the classic preference model (e.g. matrix factorization)\footnote{Here the preference model is slightly different from the original model presented in work \cite{liang2016modeling} in that we employ Bernoulli likelihood instead of Gaussian likelihood on $x_{ui}$. In fact, Bernoulli likelihood is more natural for the binary variable \cite{chen2018modeling}.} and factorized by the latent vectors $p_u$ and $q_i$, which respectively characterize latent preferences of the user $u$ and latent attributes of the item $i$. To facilitate the description, here we collect the parameters of the preference model as $\theta=\{p_u,q_i\}_{u \in U, i \in I}$. Also, we remark that it would be straightforward to replace the matrix factorization with more sophisticated models such as factorisation machines \cite{rendle2012factorization} or neural networks \cite{he2017neural}, whenever needed.
\subsection{Analysis of EXMF from variational perspective}
The marginal likelihood of EXMF is composed of a sum over the marginal likelihood of individual datapoint $\log p(X) = \sum\limits_{ui} {\log p({x_{ui}})}$, which can be rewritten as:
\begin{align}
\log p({x_{ui}}) &= {E_q}[\log p({x_{ui}},{a_{ui}}) - \log q({a_{ui}}|{x_{ui}})]\notag \\
 &+ {E_q}[\log p({a_{ui}}|{x_{ui}}) - \log q({a_{ui}}|{x_{ui}})] \notag \\
 &= L(\theta ,q;{x_{ui}}) + KL(q({a_{ui}}|{x_{ui}})||p({a_{ui}}|{x_{ui}}))
\end{align}
where $q(a_{ui}|x_{ui})$ is defined as an approximated variational posterior of $a_{ui}$.Since the second KL-divergence term is non-negative, optimizing marginal likelihood can be translated to optimize the evidence lower bound (ELBO) $L(\theta ,q;{x_{ui}})$ w.r.t. both the variational posterior and the preference parameters $\theta$. Classic variational methods \cite{hoffman2013stochastic} usually employ conjugate variational distribution and individual variational parameters\footnote{Note that the EM algorithm presented in \cite{liang2016modeling} is a special case of the classic variational inference.}, i.e. $q({a_{ui}}|{x_{ui}})=Bernoulli(\gamma_{ui})$. For convenience we collect variational parameters $\gamma _{ui}$ as matrix $Y$. Then, the ELBO can be transformed into:
\begin{align}
L(\theta ,Y ;X) &= \sum\limits_{ui} {{E_q}[\log p({x_{ui}},{a_{ui}}) - \log q({a_{ui}}|{x_{ui}})]} \notag \\
 &= \sum\limits_{ui} {{\gamma _{ui}}\ell ({x_{ui}},\sigma (p_u^ \top {q_i}))}  + \sum\limits_{ui} {f({\gamma _{ui}})} \label{eq:al}
\end{align}
The EBLO is composed of the two terms. The first term is a weighted Cross-Entropy loss for the predicted preference, where $\ell(a,b)=alog(b)+(1-a)log(1-b)$. The second term is a loss function w.r.t $\gamma_{ui}$:
\begin{align}
f({\gamma _{ui}}) =(1 - {\gamma _{ui}})\ell ({x_{ui}},\varepsilon ) + \ell ({\gamma _{ui}},{\eta_{ui}}) - \ell ({\gamma _{ui}},{\gamma _{ui}})
\end{align}

 \textbf{Exposure probability as the data confidence.} One observation from equation (\ref{eq:al}) is observed that the variational parameters $\gamma_{ui}$, which characterize the probability of the event that user $i$ is exposed to item $j$, act as the confidence weights of the corresponding data to infer the preference parameters $\theta$ ($\theta=\{p_u,q_i\}_{u \in U, i \in I}$). This is clear by considering the following fact: when $\gamma_{ui}$ becomes larger (or smaller), the inferred user and item factors $p_u,q_i$ make more (or less) contributions on the objective function. This finding is consistent with our intuitions. Only if the user has been exposed to the item, can he decide whether or not to consume the items based on his preference. Thus, the data with larger exposure are more reliable in deriving user's preference.

  \textbf{Weaknesses.} Although EXMF is capable of adaptively deriving the confidence of the data, it has two critical weaknesses: (1) Calculating gradients over all the unobserved data EXMF $(O(n\times m))$ is computational expensive and thus it practical use is limited. Although some sampling strategies can be used to speed up the algorithm, the gradient estimator exhibits high variance. Typically, in real world large datasets, each user will only be exposed to a relatively small fraction of the potential items that they could interact with. That is, the $\gamma _{ui}$ of most data are small and they make limited contribution on updating parameters $\theta$. Existing sampling strategy will usually draw uninformative data with small $\gamma_{ui}$, which deteriorates the convergence and even accuracy of the model. (2) EXMF assumes user-independent posteriors of user's exposure. On the one hand, the number of variational parameters $\gamma _{ui}$ grows quickly with the number of users and items ($n \times m$). This will pose over-fitting and efficiency problems.  On the other hand, Independent assumption of users' exposure is not practical in real world. Typically, users with directly social relations, co-purchased items or common communities will exhibit correlations in their exposure.

Thus, we are interested in, and propose a solution to two related problems:

\begin{enumerate}
\item A novel variational model of user's exposure that can both capture users' correlation and employ fewer variational parameters to speed up inference and alleviate overfitting.

\item A sampling strategy that can draw informative training instances to speed up gradient estimation and reduce sampling variance.
\end{enumerate}

\section{SamWalker}
To solve the above problems, as illustrated in Figure \ref{fg:model}, we first consider the correlations between socially connected users and propose a new social network-based recommendation method SamWalker, that replaces individual variational parameters with a social context-aware function: $Y = g_{\varphi}(X,T)$. Specifically, we design a transformation function $g_{\varphi}$ with parameters $\varphi$ that map the local social context of the user, i.e. whether his direct or indirect social friends have consumed the item, into the probability of his exposure to the item. It is reasonable since users usually get item information from social network and their exposure to items depend on their local social contexts. An idea of modeling transformation function $g_{\varphi}$ is to iteratively simulate the information spread via the social network. Similar to the PageRank algorithm \cite{page1999pagerank}, the label of user's exposure is initially set according to his consumption (\ie ${\gamma _{ui}^{(0)}}=x_{ui}$). Then, all users spread their item information to their connected friends via the social network as illustrated in Figure \ref{fg:model}. The spread process is repeated until a global stable state is achieved. In each step, users collect information from the connected social friends (neighbors) and reconstruct their exposure as follows:
\begin{align}
{\gamma _{ui}^{(t+1)}} = {(1-c)}{\gamma_{ui}^{(0)}} + \sum\limits_{k \in {\mathcal T_u}} {c\varphi_{uk}}{\gamma _{ki}^{(t)}} \label{eq:it}
\end{align}
The parameter $c$ $(0\le c \le 1)$ specifies the relative contribution from the social friends and the initial label. $\varphi_{uk}$ is defined as the edge weight, which balances the heterogenous effect from different graph neighbors ($k\in \mathcal T_u$) and meets $\sum_{k\in \mathcal T_u}\varphi_{ik}=1$. Overall, SamWalker replaces $\gamma_{ui}$ with a social context-aware function $g$ parameterized by $\varphi$, to which the equation (\ref{eq:it}) converges:
\begin{align}
Y = g_{\varphi}(X,T) \equiv\mathop {\lim }\limits_{t \to \infty } {Y^{(t)}}=  {(I - cW )^{ - 1}}(1 - c)X \label{eq:gl}
 \end{align}
 where we collect variables $\gamma _{ui}^{(t)}$ for every user-item pairs $(u,i)$ as a matrix $Y^{(t)}$. Also, we collect $\varphi_{uk}$ as a matrix $W$, in which $W_{uk}=\varphi_{uk}$ for connected user pairs and $W_{uk}=0$ for others. As we can see from equation (\ref{eq:gl}), SamWalker replaces the posterior expectation of user's exposure with a weighted combination of the users' consumption in his social network. The weight matrix $(I - cW )^{ - 1}$ is a graph or diffusion kernel \cite{zhou2004learning}, which has been widely adopted to measure node proximity in the network and depends on the edge weight parameters $\varphi$ for every social ties. The inference of user's exposure can benefit from the knowledge of his similar social friends. Overall, SamWalker is capable of capturing social correlations between users and reduces the number of variational parameters from $O(n\times m)$ to $O(|E|)$, where $|E|$ denotes the number of edges in the social network. By iteratively learning the transformation function and the preference model, SamWalker can adaptively specify different weights to different data based on users' social contexts.



\section{SamWalker++}
A key limitation of SamWalker is that it requires additional social network data to model users' correlation, which may not be easy to collect in many recommender systems. To deal with this problem, we further propose SamWalker++, which does not use any side information. SamWalker++ models data confidence with a constructed pseudo-social network. As illustrated in Figure \ref{jiashiyi}, in the pseudo-social network similar users are connected with specific additional nodes so that the learning of a user's exposure can benefit from the information of his similar users. We introduce two kinds of nodes:
\begin{itemize}
\item \textbf{Item nodes}: Note that the users who have interacted with common items tend to have similar exposure. Thus, we introduce item nodes as bridges and link the user-item pairs with positive feedback, so that the inference of one's exposure can benefit from the knowledge of his similar users with co-purchased items.
\item \textbf{Community nodes:} Motivated by the social psychology statement \cite{palla2005uncovering,zhou2011understanding} that users are clustered into some content-sharing communities, we deduce the latent communities among users and introduce community nodes as medium to transfer the knowledge of users' exposure along the community. That is, as shown in Figure \ref{jiashiyi}, we give links for the user-community pairs if the user belongs to the community. Note that pre-computing users' community with existing community discovering algorithms is not-optimal, as the rich supervised signals from users' exposure have not been exploited. Thus, we prefer an end-to-end model. We first connect each user-community pair and initialize the community distribution for each user as an uniform distribution. We then adaptively update users' community distribution by optimizing the objective function with the training process going on.
\end{itemize}

 Given a pseudo-social network, we devise a novel exposure model on the network and specify data confidence weights with a network-aware function $Y = {h_\varphi }(X,G)$, where $G$ denotes the constructed pseudo-social network encoding similarity among users. That is, we design a transformation function $h$ with parameters $\varphi$ that maps the behaviors of the target user and his similar users into the probability of his exposure to the item. This way, the rich knowledge from these similar users can be transferred to learn the target user, which mitigates over-fitting problem and boosts inference accuracy. Similar to SamWalker, a promising way of modeling transformation function $h_{\varphi}$ is to iteratively simulate the knowledge flowing along the pseudo-social network. We initially set the label of users' exposure with his consumption and then reconstruct exposure with the information from their connected users. We model two kinds of information propagation.

\textbf{Along item nodes.} On the one hand, the users with co-purchased items provide knowledge on the target user's exposure. We define the message from this kind of similar users to the target user $u$ along the co-purchased items as:
\begin{align}
m_{uj}^{u \leftarrow i \leftarrow u} = \sum\limits_{i \in {N_u} \cap I} {\sum\limits_{v \in {N_i}} {\varphi _{ui}^{u \leftarrow i}\varphi _{iv}^{i \leftarrow u}\gamma _{vj}^{(t)}} } \label{eq:ch1}
\end{align}

where $N_u$ and $N_i$ denote the neighbor nodes sets of the user $u$ and item $i$; $\varphi _{ui}^{u \leftarrow i}, \varphi _{iv}^{i \leftarrow u}$ denote edge weights, balancing the contributions of information from different edges and meeting $\sum\limits_{i \in {N_u} \cap I} {\varphi _{ui}^{u \leftarrow i}}  = 1,\sum\limits_{v \in {N_i}} {\varphi _{iv}^{i \leftarrow u}} = 1$. The product of $\varphi _{ui}^{u \leftarrow i}$ and $\varphi _{iv}^{i \leftarrow u}$ can be intepreted as path strength for $u \leftarrow i \leftarrow v$, characterizing the strength of information flowing from the user $v$ to the target user $u$. As we can see, the users with more and stronger paths, suggesting that they exhibit more similarity with the target user, will bring more information on learning.

\textbf{Along community nodes.} On the other hand, users belonging to common communities will also exhibit correlations in their exposure. We define the message from this kind of similar users as:
\begin{align}
m_{uj}^{u \leftarrow c \leftarrow u} = \sum\limits_{c \in {N_u} \cap C} {\sum\limits_{v \in {N_c}} {\varphi _{uc}^{u \leftarrow c}\varphi _{cv}^{c \leftarrow u}\gamma _{vj}^{(t)}} } \label{eq:ch2}
\end{align}
where ${\varphi _{uc}^{u \leftarrow c}}$ and ${\varphi _{cv}^{c \leftarrow u}}$ denote edge weights, balancing effect of different information edges and meeting $\sum\limits_{c \in {N_u} \cap C} {\varphi _{uc}^{u \leftarrow c}}  = 1,\sum\limits_{v \in {N_c}} {\varphi _{cv}^{c \leftarrow u}} = 1$. Intuitively, $\varphi _{uc}^{u \leftarrow c}$ can be interpreted as user's community distribution and $\varphi _{cv}^{c \leftarrow u}$ as the extent to which the community member $v$ exposes to the community $c$. The inference of user $u$'s exposure can refer to the exposure of other community members. The knowledge flows along the path $u \leftarrow c \leftarrow v$ with strength $\varphi _{uc}^{u \leftarrow c}\varphi _{cv}^{c \leftarrow u}$.

\textbf{Aggregation.} We now aggregate the two messages to refine the target user $u$'s exposure:
\begin{align}
\gamma _{uj}^{(t + 1)} = c({b_u}m_{uj}^{u \leftarrow i \leftarrow u} + (1 - {b_u})m_{uj}^{u \leftarrow c \leftarrow u}) + (1-c)\gamma _{uj}^{(0)} \label{eq:ch3}
\end{align}
The parameter $c$ $(0\le c \le 1)$ specifies the relative contributions from the initial label and the connected\footnote{Here we define the connected users as the users with common items or common communities.} similar users. $b_u$ balances the contributions of these two messages. Moreover, with the transitivity of similarity, our similar users of similar users, or even higher-order similar users, are potentially similar to us. These high-order similar users, as abundant information resources, usually provide useful knowledge of the target user's exposure. Thus, referring to SamWalker, we stack multi-stage of refinement as illustrated in Figure \ref{fg:model} so that the inference of a user's exposure can benefit from high-order connected users. Overall, SamWalker++ replace $\gamma_{ui}$ with a network-aware function $h$ parameterized by $\varphi  = \{ {\varphi ^{u \leftarrow i}},{\varphi ^{i \leftarrow u}},{\varphi ^{u \leftarrow c}},{\varphi ^{c \leftarrow u}},b\}$, to which equations (\ref{eq:ch1}),(\ref{eq:ch2}),(\ref{eq:ch3}) converge:
\begin{align}
Y = {h_\varphi }(X,G) \equiv \mathop {\lim }\limits_{t \to \infty } {Y^{(t)}} = {(I - cW^+)^{ - 1}}(1 - c)X
\end{align}
where $W^+ = B{\Phi ^{u \leftarrow i}}{({\Phi ^{i \leftarrow u}})^T} + (1 - B){\Phi ^{u \leftarrow c}}{({\Phi ^{c \leftarrow u}})^T}$. Also, we collect parameters $b_u$ for each user as a diagonal matrix $B$ and collect ${\varphi_{ui} ^{u \leftarrow i}}$ as a matrix ${\Phi ^{u \leftarrow i}}$, in which $\Phi _{ui}^{u \leftarrow i} = \varphi _{ui}^{u \leftarrow i}$ for connected user-item pair $(u,i)$ and $\Phi _{ui}^{u \leftarrow i} =0$ for others. Similar treatments are used for parameters ${\varphi ^{i \leftarrow u}},{\varphi ^{u \leftarrow c}},{\varphi ^{c \leftarrow u}}$.  SamWalker++ models user's exposure with a weighted combination of other users' consumption. Also, the weight matrix $(I - cW^+ )^{ - 1}$ is a graph or diffusion kernel \cite{zhou2004learning} characterizing users proximity in the pseudo-network, which naturally encodes similarity or even high-order similarity between users into the inference procedures, which boosts the inference accuracy.

 \begin{figure}[t!]
\centering
\includegraphics[width=0.49\textwidth]{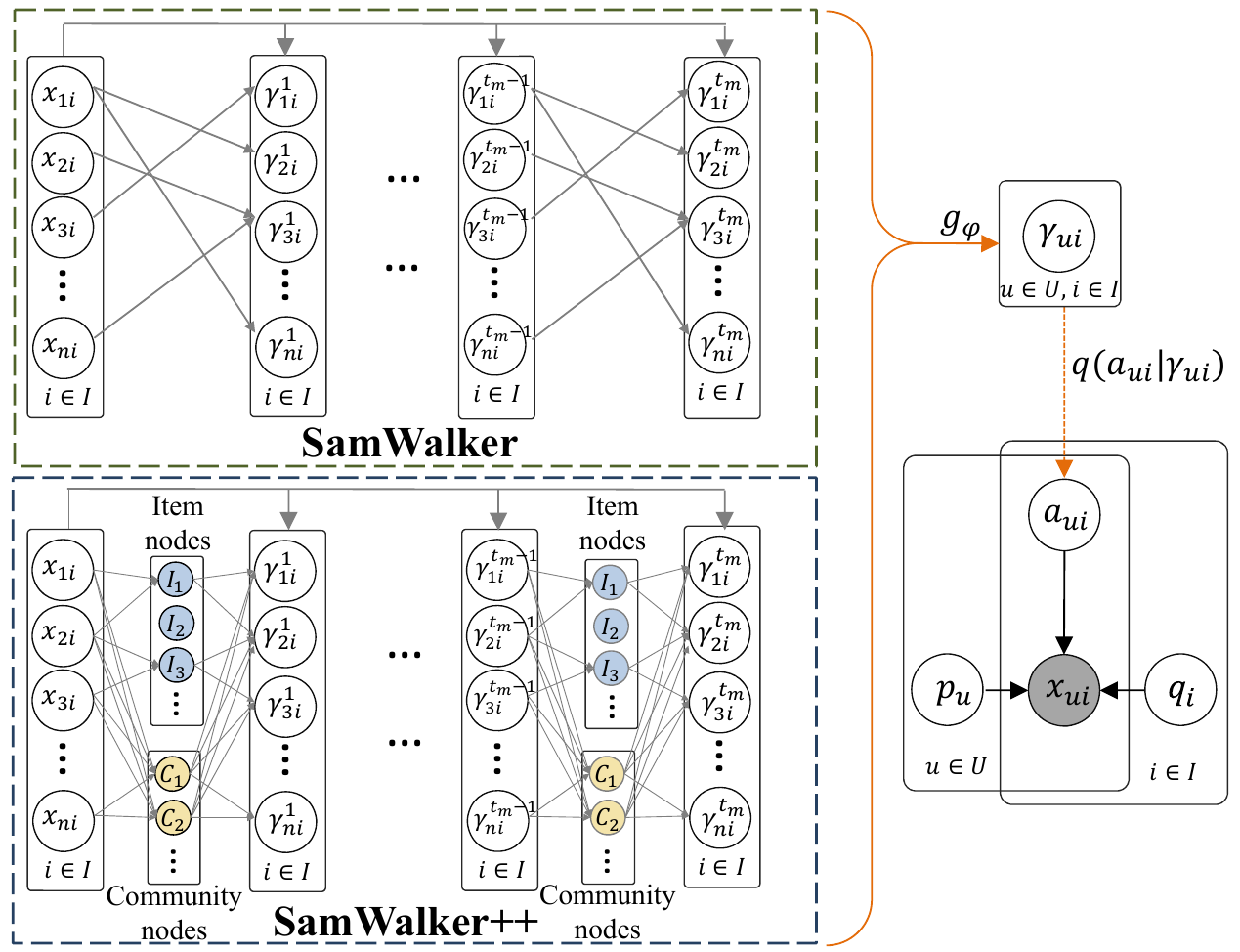}
 \caption{A schematic view of the proposed SamWalker (Left-upper) and SamWalker++ (Left-bottem). SamWalker directly simulates the information spread via the social network, while SamWalker++ transfers the knowledge between users by introducing specific item and community nodes as bridges.}
  \vspace{-0.3cm}
\label{fg:model}
\end{figure}

\subsection{Discussion}
The proposed SamWalker++ satisfies four desirable properties:

\textbf{Side-information free.} As we can see, SamWalker++ only uses implicit feedback data and does not require any side information (e.g. social network, item contents, tags). As side information are not available in many recommender system, SamWalker++ can be applied in more situations comparing with the methods using side information(e.g. SamWalker).

\textbf{Mitigate over-fitting.} Another advantage of SamWalker++ is its ability to mitigate over-fitting problem. One evidence supporting this point can be seen from the less parameters of SamWalker++ comparing with EXMF. SamWalker++ reduces the number of variational parameters from $O(nm)$ to $O(|X^+|+nK)$, where $|X^+|$ denotes the number of observed positive feedback in the dataset and $K$ denotes the number of inferred communities. Due to the sparsity of the implicit feedback data, the number of positive data ($|X^+|$) is much less than the the number of all data ($nm$).

How SamWalker++ mitigates over-fitting can be intepreted from another perspective. Referring to the analyses presented in \cite{chen2020fast}, let us draw an analogy with the floating balls in the water, as illustrated in Figure \ref{fg:over}. Learning exposure-based recommendation model according to equation (\ref{eq:al}) will give a force to pull up these positive balls (data) and push down these unobserved balls (data). For the vanilla EXMF model, the data confidence weights will easily achieve extreme values (${\gamma _{ui}} \approx 1$ for the positive data and ${\gamma _{ui}} \approx 0$ for the unobserved data), where the unobserved data make little contribution to training the recommendation model and the model will suffer from over-fitting. But in SamWalker++, users' exposure (data confidence) are connected with additional items or communities, which can be analogies as additional balls with elastic links connecting the data. Naturally, the unobserved data with more and stronger connections with positive data, will be pulled up higher due to the force from the links. This way, when the model has well fitted the data, the positive and the unobserved ball(data) will get stable at different depth in water, as the knowledge (force) will prorogation among the data. The over-fitting effect will be mitigated.

\textbf{Adaption.} The data confidence is defined with a parameterized function instead of pre-defined values. Thus, the data confidence will adaptively evolve with training process going on, which is more flexible and does not require manual tuning of confidence weights. Moreover, as SamWalker++ integrates other user information, some irrelevant factors and even data noises may be injected into the inference. Fortunately, SamWalker++ is trained in a supervised manner by maximizing data likelihood so that the model can adaptively recognize important edges and extract useful information from the network.

\textbf{Fast informative sampling.} Also, SamWalker++ supports fast informative sampling, which will be discussed in the next section.

\begin{figure}[t!]
\centering
\includegraphics[width=0.49\textwidth]{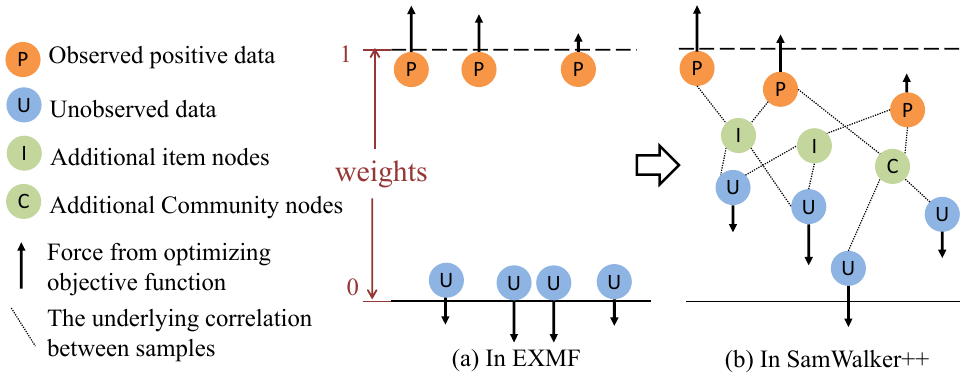}
 \caption{Illustration of how SamWalker++ mitigates over-fitting}
\label{fg:over}
\end{figure}

\section{Inference with informative sampler}
\subsection{Random walker-based sampler}
\label{sap}
Stochastic gradient descent (SGD), as a promising solution to speed up training procedures, has been widely applied in recommendation. The sampling strategy plays an important role in SGD, as it determines which data are used to update parameters and how often. However, since the informative instances with larger confidence $\gamma_{ui}$ are usually buried in a large pile of uninformative ones, existing samplers usually fail to pick out informative data, leading to poor convergence and non-optimal performance. Thus, in this section, we develop a novel informative sampling strategy and address the following two key research questions: (Q1) Given the current learned data confidence $\gamma_{ui}$, how to define the informative sampling distribution? (Q2) Given informative sampling distribution, how to draw instances efficiently?

For the question (Q1), intuitively, the informative data with larger confidence $\gamma_{ui}$ should be sampled with larger probability, since these terms make more contribution to the objective function. In fact, we have the following lemmas:
\newtheorem{lem}{Lemma}
\begin{lem}
\label{la1}
To evaluate the unbiased gradient of $L$ w.r.t $\theta$, the sampling strategy with distribution $p_{ui} \propto \gamma_{ui}$ can reduce sampling variance.
\end{lem}
The proof is presented in appendix.

\begin{lem}
\label{la2}
To evaluate the unbiased gradient of $L$ w.r.t $\theta$, the sampling strategy with distribution $p_{ui} \propto \gamma_{ui}$ can speed up gradient calculation.
\end{lem}

\begin{proof}
If the sampling distribution is proportional to the data confidence ($p_{ui} =\gamma_{ui}/Z$, where $Z = \sum\limits_{u \in U,i \in I} {{\gamma _{ui}}} $), we have:
\begin{align}
\frac{{\partial L}}{{\partial \theta }}{\rm{ = }}\sum\limits_{u \in U,i \in I} {{\gamma _{ui}}\frac{{\partial {\ell _{ui}}}}{{\partial \theta }}}  = \sum\limits_{u \in U,i \in I} {\frac{{{p_{ui}}}}{Z}\frac{{\partial {\ell _{ui}}}}{{\partial \theta }}} {\rm{ = }}E_p[\sum\limits_{(a,b) \in p} {\frac{1}{Z}\frac{{\partial {\ell _{ab}}}}{{\partial \theta }}} ] \label{eq:sa}
\end{align}
where the confidence weights $\gamma_{ui}$ have been absorbed into the sampling bias and does not need calculating in each iteration, which saves much time.
\end{proof}

The question (Q2) is more challenging, as the sampling distribution will evolve over a large data space as training process going on. A naive implementation of informative sampler is to estimate and rank the current learned confidence weights $\gamma_{ui}$ for every user-item pairs and then pick out the informative data based on $\gamma_{ui}$. It is apparently inefficient and can not satisfy practical requirement. To avoid estimating confidence weights, we propose the following sampling strategy for our SamWalker and SamWalker++:

\textbf{Random Walk-based sampling strategy.} For the target user $u$, we perform the random walk along the network from user node $u$ to sample the informative feedback data of user $u$. At each step $t$ of random walk, supposing we are at a certain user $v$, we have two options:

(1) With probability $c$, we terminate the random walk. We stay at user $v$ and randomly (uniformly) select a portion of ($N_v/\beta$) items that have been consumed by the user $u$, where $N_v$ denotes the number of items consumed by the user $v$. Then we add the feedback data of user $u$ on these selected items into sampled set $S$.

(2) With probability $(1-c)$, we continue our random walk. For the SamWalker model, we randomly select one of $v$'s connected friends $r$ based on personalized tie strength $\varphi_{vr}$ and walk to $r$ for the next walk step; For the SamWalker++ model, we first flip a coin based on $a_v$ to decide which kinds of nodes (items or communities) we would like to walk along. If we choose item nodes (or community nodes) as a medium, then we randomly walk to one of $v$'s connected items $i$ (or communities $c$) based on edge weights ${\varphi _{vi}^{u \leftarrow i}}$ (${\varphi _{vc}^{u \leftarrow c}}$). After that, we randomly walk from $i$ (or $c$) to its connected user node $r$ based on ${\varphi _{ir}^{i \leftarrow u}}$ (or ${\varphi _{cr}^{c \leftarrow u}}$) for the next walk step.

Our random walk-based sampling strategy satisfies the following desirable property:
\begin{lem}
\label{la3}
The sampling probability of the above random walk-based strategy is proportional to the data confidence.
\end{lem}

\begin{proof}
It is easy to check that transformation probability from one user $u$ to another user $r$ in the step (2) is consistent with the $(u,r)$-th element of the matrix $W$ for SamWalker or $W^+$ for SamWalker++. Further we can find that the $(u,r)$-th element of matrix ${(cW^+ )^t}(1-c)$ (or ${(cW )^t}(1-c)$) is the probability of starting from the source user $u$ and terminating at the user $r$ in step $t$. Correspondingly, the $(u,i)$-th element of matrix ${(c\Phi )^t}(1-c)X$ represents the sampled probability of the user-item feedback data $x_{ui}$ in the step $t$. Sum over the probability in different steps, we have the sampled probability of the data as follow:
\begin{align}
P = \sum\limits_{t = 0}^\infty  {{{(c\Phi )}^t}(1 - c)X/\beta }  \propto Y
 \end{align}
which is proportional to the data confidence.
\end{proof}

Here we give a more intuitive explanation of our proposed random walk-based sampling strategy. The random walk from the target user $u$ will explore user's network and finally randomly arrive at a specific user $v$ based on their graph proximity. Note that the network is constructed with user social relations, communities, or co-purchased items. Higher edge weights or shorter distance between the user $v$ to the target user $u$, indicates more similarity between the two users and thus $v$ will be selected with higher probability. The corresponding items, which are exposed (consumed) to $u$'s similar users, are more likely exposed to the user $u$. Our random walk-based sampling strategy encodes the knowledge from other similar users and thus is capable of selecting informative instances.

In practice, we usually conduct $\alpha$ times random walk for each user to achieve more reliable mini-batch stochastic optimization. The parameters $\alpha, \beta$ control the batch size. Note that there is a chance for a single random walk to continue forever. In fact, we pay more attention to user's local social context and thus terminate the random walk when the number of steps exceeds a certain threshold ($t>t_m$). Concretely, when $t>t_m$, we uniformly walk to a random user in the system and sample the data as option (1).

\subsection{Inference of the edge weights $\varphi$}

Note that knowledge transfer between the data may also inject some irrelevant factors or even noises. Thus, we would like to train SamWalker and SamWalker++ in a supervised manner so that the model can adaptively recognize important edges and extract useful information along the network. We achieve this by optimizing the lower bound of margin likelihood (equation \ref{eq:al}) w.r.t parameters $\varphi$ with stochastic gradient methods. However, directly deriving gradient from transformation function $g_\varphi$ (equation (\ref{eq:gl})) involves matrix inversion and suffers from low efficiency. Alteratively, as illustrated in Figure \ref{fg:model}, we iteratively simulates information spread as equation (\ref{eq:it}), and stacks multi-layers neural network to infer the personalized edge weights. We also reparameterize $\varphi _{uv},{\varphi _{ui}^{u \leftarrow i}},{\varphi _{iv}^{i \leftarrow u}},{\varphi _{uc}^{u \leftarrow c},\varphi _{cv}^{c \leftarrow u}}$ with a Softmax transformation to deal with sum-to-one constraints. Backward prorogation can be easily conducted to infer tie strength $\varphi$, without requiring time-consuming matrix inversion. Further, mini-batch-based stochastic gradient methods can be employed to speed up the inference. Note that the data set $S$ from random walk strategy is sampled for updating the recommendation model, and may not be suitable for $\varphi$. Thus, we choose the uniform sampler. In each step, we randomly (uniformly) select a portion of $N_{SI}$ items and update tie strength based on users' exposure on these selected items. Overall, the inference of our SamWalker and SamWalker++ is presented in Algorithm \ref{al}.

\begin{algorithm}[t]
\footnotesize
\caption{Inference of  SamWalker and SamWalker++}
\begin{algorithmic}[1]
\label{al}
\STATE Initialize parameters randomly;
\WHILE {not converge}

\STATE Sample a set of data $S$ based on the random walk strategy as mentioned in section 6.
\STATE Update parameters $\theta$ of the preference model based on the estimated gradient on the sampled data [equation (\ref{eq:sa})].
\STATE Randomly select a portion of $N_{SI}$ items.
\STATE Update parameters $\varphi$ based on backward propagation along the neural network for the selected items [equations (\ref{eq:it},\ref{eq:ch1},\ref{eq:ch2},\ref{eq:ch3})].

\ENDWHILE
\end{algorithmic}
\end{algorithm}

\subsection{Complexity Analysis}

The time complexity of the inference of SamWalker and SamWalker++ is attributed to the following three parts: (1) In sampling step, we will conduct $\alpha$ times random walk for each user to generate sampled data set $S$. The time for this part is $O(\alpha n t_m +|E| +|S|)$, where $n$ denotes the number of users in the system and $|E|$ denotes the number of edges in the network; $t_m$ denotes the max depth of random walk and $|S|$ denotes the number of data in the set $S$. (2) When inferring the preference parameters $\theta$, we just estimate gradients on the sampled data $S$. The time for this step is $O(|S|d)$. (3) When inferring the parameters $\varphi$ of the transformation function, we conduct the gradient back propagation along the network for the selected $N_{SI}$ items. The complexity for this part is $O(N_{SI}|E|)$. Hence, the overall computational complexity is $O(\alpha n t_m+|S|d+N_{SI}|E|)$. Due to the sparsity of the recommendation data, users usually have limited social friends, interactions and communities. Note that $|E|=|T^+|$ for SamWalker and $|E|=nD+|X^+|$ for SamWalker++, where $|T^+|$ denotes the number of social relations and $|X^+|$ denotes the number of observed positive feedback. Also, similar to many recent works \cite{wu2016collaborative,rendle2009bpr} we usually let $|S|$ be five times as large as the number of observed data and let the number of selected items $N_{SI}$ be 100. Thus, our algorithm is efficient on sparse implicit feedback data.

\section{Experiments and analysis}

Our experiments are intended to address the following questions:
\begin{enumerate}[(Q1)]
\item  Do SamWalker and SamWalker++ outperform state-of-the-art recommendation methods?
\item How does SamWalker++ compare with SamWalker?
\item  How does the proposed sampling strategy perform?
\item  Is it beneficial to introduce item nodes and community nodes to infer the data confidence?
\item  How does the parameter $t_m$ (the max depth of prorogation) affect the recommendation performance?
\end{enumerate}
\subsection{Experimental protocol}
\label{se:da}

\textbf{Datasets.} Five datasets Epinions\footnote{\url{http://www.trustlet.org/epinions}}, Ciao\footnote{\url{http://www.cse.msu.edu/~tangjili/trust}}, LastFM\footnote{\url{https://grouplens.org/datasets/hetrec-2011/ }}, Moivelens-1M\footnote{\url{https://grouplens.org/datasets/movielens/}} and BookCrossings\footnote{\url{https://grouplens.org/datasets/book-crossing/}} are used in our experiments. These datasets contain users' feedback on the items. The datasets Epinions, Ciao and LastFM also contain users' social relations. The dataset statistics are presented in Table \ref{tb:da}. Similar to \cite{xiao2017learning,he2016fast}, we preprocess the datasets so that all items have at least three interactions and "binarize" user's feedback into implicit feedback. That is, as long as there exists some user-item interactions (ratings or clicks), the corresponding implicit feedback is assigned a value of 1. We also drop out items that have been consumed by too many (larger than 100) or too few (smaller than 3) users to moderate the popularity biases \cite{canamares2018should} in estimation. After that, we randomly choose 80\% of interactions for training and leave the remainder for testing.

\textbf{Hyper-parameters.} Grid search and 5-fold cross validation are used to find the best parameters. Also, we choose Adam as our optimizer. In our SamWalker and SamWalker++, we
we set $\eta_{ij}=0.5$, $\alpha=100$, $\beta=20$, $c=0.9$, and test $t_m$ of the search space $\{3,4,5,6\}$, learning rate of $\{0.001,0.01,0.1,1\}$ and decay of $\{0.0001,0.001,0.01,0.1,1\}$.  The setting of compared methods are referring to related works or validated in our experiments. All experiments are conducted on a server with 2 TiTanX GPUs, Intel E5-2620 CPUs and 256G RAM\footnote{Source code will be available at github \url{https://github.com/jiawei-chen/SamWalker}}.

\textbf{Compared methods.} The compared methods are as follows. Table \ref{tb:ch} also summarizes their characteristics.
\begin{itemize}
\item WMF(ALS)  \cite{hu2008collaborative}: The classic weighted matrix factorization model for implicit feedback data with memorization-based fast learning algorithm.
\item BPR \cite{rendle2009bpr}: The classic pair-wise method for recommendation, coupled with matrix factorization. For efficient learning, BPR employs uniform sampling strategy to draw the training instances.
\item EXMF \cite{liang2016modeling}: A probabilistic model that directly incorporates user's exposure to items into traditional matrix factorization. We refer to \cite{liang2016modeling} and choose an item-dependent prior of user's exposure.
\item FAWMF \cite{chen2020fast}: A fast matrix factorization model with adaptive confidence weights and memorization-based learning algorithm.
\item LightGCN \cite{he2020lightgcn}: State-of-the-art recommendation model with graph neural network on the user-item interaction graph.
\item SBPR\cite{zhao2014leveraging}: SBPR integrates social information into BPR by assuming that the items consumed by connected friends are ranked higher than those not.
\item SERec-Bo\cite{DBLP:conf/aaai/WangZYZ18}: A probabilistic model that extends the EXMF model with social influence on user's exposure. Here we choose SERec-Bo as a comparison since it performs better than SERec-Re.
\item SoEXBMF \cite{chen2018modeling}:  A probabilistic model that further extends the EXMF model with both social knowledge influence and social consumption influence.
\end{itemize}

\begin{table}[t!]
\centering
\tiny
\caption{Statistics of five datasets.}
\label{tb:da}
\begin{tabular}{|c|c|c|c|c|c|}
\hline
Datasets     & \begin{tabular}[c]{@{}c@{}}Number\\ of users\end{tabular} & \begin{tabular}[c]{@{}c@{}}Number\\ of items\end{tabular} & \begin{tabular}[c]{@{}c@{}}Number of \\ interactions\end{tabular} & \begin{tabular}[c]{@{}c@{}}Sparsity of\\ interactions\end{tabular} & \begin{tabular}[c]{@{}c@{}}Number of\\ social relations\end{tabular} \\ \hline
LastFM       & 1,892                                                     & 4,489                                                     & 52,668                                                            & 0.62\%                                                             & 25,434                                                               \\ \hline
Ciao         & 5,298                                                     & 19,301                                                    & 138,840                                                           & 0.14\%                                                             & 106,640                                                              \\ \hline
Epinions     & 21,290                                                    & 34,075                                                    & 333,916                                                           & 0.05\%                                                             & 414,549                                                              \\ \hline
Moivelens-1M & 6,040                                                     & 3,678                                                     & 1,000,177                                                         & 4.50\%                                                             & None                                                                 \\ \hline
BookCrossing & 13,097                                                    & 37,075                                                    & 473,846                                                           & 0.10\%                                                             & None                                                                 \\ \hline
\end{tabular}
\end{table}

\begin{table}[t!]
\caption{The characteristics of the compared methods.}
\center
\tiny
\label{tb:ch}
\begin{tabular}{|c|c|c|c|c|}
\hline
Methods     & Social?      & \begin{tabular}[c]{@{}c@{}}Exposure\\ -based?\end{tabular} & Sampling?    & Complexity                       \\ \hline
WMF(ALS)    & $\backslash$ & $\backslash$                                               & $\backslash$ & $O((n+m)d^3)$                    \\ \hline
BPR         & $\backslash$ & $\backslash$                                               & ${\surd}$    & $O((n+m+|S|)d)$                   \\ \hline
EXMF        & $\backslash$ & ${\surd}$                                                  & $\backslash$ & $O(nmd)$                         \\ \hline
FAWMF       & $\backslash$ & ${\surd}$                                                  & $\backslash$ & $O(|\mathbf{X}^+|(K+D))$         \\ \hline
LightGCN    & $\backslash$ & $\backslash$                                               & ${\surd}$    & $O((|S|+|\mathbf{X}^+|)d)$             \\ \hline
SBPR        & ${\surd}$    & $\backslash$                                               & ${\surd}$    & $O((n+m+|S|)d)$                  \\ \hline
SERec-Bo    & ${\surd}$    & ${\surd}$                                                  & $\backslash$ & $O(nmd)$                         \\ \hline
SoEXBMF     & ${\surd}$    & ${\surd}$                                                  & $\backslash$ & $O(nmd^2)$                       \\ \hline
SamWalker   & ${\surd}$    & ${\surd}$                                                  & ${\surd}$    & $O(\alpha n t_m+|S|d+N_{SI}|E|)$ \\ \hline
SamWalker++ & $\backslash$ & ${\surd}$                                                  & ${\surd}$    & $O(\alpha n t_m+|S|d+N_{SI}|E|)$ \\ \hline
\end{tabular}
\end{table}

\textbf{Evaluation Metrics.} We adopt the following metrics:
\begin{itemize}
\item Recall@K (Rec@K): This metric quantifies the fraction of consumed items that are in the top-K ranking list sorted by their estimated rankings. For each user $i$, we define $Rec(i)$ as the set of recommended items in top-K and $Con(i)$ as the set of consumed items in test data for user $i$. Then we have:
\begin{align}
    Recall@K&=\frac{1}{{|U|}}\sum\limits_{i \in U} {\frac{{|Rec(i)\cap Con(i)|}}{|Con(i)|}}
\end{align}
\item Precision@K (Pre@K): This measures the fraction of the top-K items that are consumed by the user:
 \begin{align}
    Precision@K&=\frac{1}{{|U|}}\sum\limits_{i \in U} {\frac{{|Rec(i)\cap Con(i)|}}{|Rec(i)|}}
\end{align}
\item Normalized Discounted Cumulative Gain (NDCG): it measures the quality of ranking:
\begin{align}
NDCG &= \frac{1}{{|U|}}\sum\limits_{i \in U} {\frac{{DCG_i}}{{{IDCG_i}}}}
\end{align}
where ${DC{G_i}}$ is defined as follow and ${{IDCG_i}}$ is the ideal value of ${{DCG_i}}$ coming from the best ranking.
\begin{align}
{DCG_i} &= \sum\limits_{j \in  Con(i)} {\frac{1}{{{{\log }_2}(P_{ij} + 1)}}}
\end{align}
where ${P_{ij}}$ represents the rank of the item $j$ in the recommended list of the user $i$.
\item Mean Reciprocal Rank (MRR): Referring to \cite{shi2012climf}, given the ranking lists, MRR is defined as follow:
\begin{align}
MRR = \frac{1}{{|U|}}\sum\limits_{i \in U} {\sum\limits_{j \in Con(i)} {\frac{1}{{{P_{ij}}}}\prod\limits_{k \in Con(i)} {\mathbf I\left( {{P_{ik}} \ge {P_{ij}}} \right)} } }
\end{align}
\end{itemize}

\begin{table*}[t!]
\centering
\scriptsize
\caption{The performance metrics of the compared methods. The mark '*' denotes the winner in that row, while the boldface font denotes the winner among the non-social recommendation methods.  The column `Impv1' indicates the relative performance gain of our SamWalker compared to the best results among all baselines, while the column `Impv2' indicates the relative performance gain of our SamWalker++ compared to the best results among the non-social methods. }
\label{tb:re}
\begin{tabular}{|c|c|c|c|c|c|c|c|c|c|c|c|c|c|}
\hline
\begin{tabular}[c]{@{}c@{}}Data-\\ sets\end{tabular}                      & Metrics & WMF    & BPR    & EXMF   & FAWMF            & \begin{tabular}[c]{@{}c@{}}Light-\\ GCN\end{tabular} & SBPR         & SeRec        & \begin{tabular}[c]{@{}c@{}}So-\\ EXBMF\end{tabular} & \begin{tabular}[c]{@{}c@{}}Sam-\\ Walker\end{tabular} & Impv1        & \begin{tabular}[c]{@{}c@{}}Sam-\\ Walker++\end{tabular} & Impv2   \\ \hline
\multirow{4}{*}{LastFM}                                                   & Pre@5   & 0.0928 & 0.1004 & 0.0957 & 0.1011           & 0.1051                                               & 0.0956       & 0.1018       & 0.1108                                              & 0.1177*                                               & 6.22\%       & \textbf{0.1099}                                         & 4.61\%  \\ \cline{2-14}
                                                                          & Rec@5   & 0.0841 & 0.0888 & 0.0859 & 0.0906           & 0.0934                                               & 0.0851       & 0.0907       & 0.1014                                              & 0.1072*                                               & 5.68\%       & \textbf{0.0983}                                         & 5.20\%  \\ \cline{2-14}
                                                                          & NDCG    & 0.3364 & 0.3485 & 0.3477 & 0.3242           & 0.3533                                               & 0.3405       & 0.3509       & 0.3617                                              & 0.3634*                                               & 0.48\%       & \textbf{0.3601}                                         & 1.93\%  \\ \cline{2-14}
                                                                          & MRR     & 0.2596 & 0.2601 & 0.2502 & 0.2643           & 0.2669                                               & 0.2553       & 0.2661       & 0.2932                                              & 0.2992*                                               & 2.05\%       & \textbf{0.2939}                                         & 10.15\%  \\ \hline
\multirow{4}{*}{Ciao}                                                     & Pre@5   & 0.0172 & 0.0144 & 0.0095 & 0.0143           & 0.0151                                               & 0.0156       & 0.0118       & 0.0181                                              & 0.0182*                                               & 0.25\%       & \textbf{0.0184}                                         & 6.98\%  \\ \cline{2-14}
                                                                          & Rec@5   & 0.0123 & 0.0125 & 0.0105 & 0.0092           & 0.0136                                               & 0.0124       & 0.0124       & 0.0152                                              & 0.0167*                                               & 10.38\%      & \textbf{0.0161}                                         & 18.25\% \\ \cline{2-14}
                                                                          & NDCG    & 0.1757 & 0.1759 & 0.1747 & 0.1641           & 0.1800                                               & 0.1774       & 0.1770       & 0.1827                                              & 0.1811                                                & -0.84\%      & \textbf{0.1834*}                                        & 1.89\%  \\ \cline{2-14}
                                                                          & MRR     & 0.0541 & 0.0464 & 0.0377 & 0.0446           & 0.0556                                               & 0.0478       & 0.0395       & 0.0587                                              & 0.0588                                                & 0.24\%       & \textbf{0.0609*}                                        & 9.45\%  \\ \hline
\multirow{4}{*}{Epinions}                                                 & Pre@5   & 0.0095 & 0.0087 & 0.0079 & 0.0090           & 0.0098                                               & 0.0088       & 0.0073       & 0.0119                                              & 0.0149                                                & 24.43\%      & \textbf{0.0165*}                                        & 67.72\% \\ \cline{2-14}
                                                                          & Rec@5   & 0.0087 & 0.0096 & 0.0093 & 0.0071           & 0.0116                                               & 0.0089       & 0.0101       & 0.0126                                              & 0.0184                                                & 46.26\%      & \textbf{0.0186*}                                        & 60.37\% \\ \cline{2-14}
                                                                          & NDCG    & 0.1522 & 0.1541 & 0.1517 & 0.1444           & 0.1593                                               & 0.1542       & 0.1577       & 0.1600                                              & 0.1656                                                & 3.45\%       & \textbf{0.1693*}                                        & 6.24\%  \\ \cline{2-14}
                                                                          & MRR     & 0.0341 & 0.0294 & 0.0288 & 0.0295           & 0.0426                                               & 0.0331       & 0.0287       & 0.0422                                              & 0.0506*                                                & 18.84\%      & \textbf{0.0505}                                        & 18.68\% \\ \hline
\multirow{4}{*}{\begin{tabular}[c]{@{}c@{}}Movie-\\ lens-1M\end{tabular}} & Pre@5   & 0.3841 & 0.3613 & 0.3871 & \textbf{0.4054*} & 0.4008                                               & $\backslash$ & $\backslash$ & $\backslash$                                        & $\backslash$                                          & $\backslash$ & 0.3929                                                  & -3.08\% \\ \cline{2-14}
                                                                          & Rec@5   & 0.0924 & 0.0798 & 0.0936 & \textbf{0.0949*} & 0.0948                                               & $\backslash$ & $\backslash$ & $\backslash$                                        & $\backslash$                                          & $\backslash$ & 0.0943                                                  & -0.63\% \\ \cline{2-14}
                                                                          & NDCG    & 0.5971 & 0.5814 & 0.5963 & 0.5911 & \textbf{0.5976*}                                               & $\backslash$ & $\backslash$ & $\backslash$                                        & $\backslash$                                          & $\backslash$ & 0.5920                                                  & -0.94\% \\ \cline{2-14}
                                                                          & MRR     & 0.6145 & 0.5779 & 0.6150 & 0.6314 & \textbf{0.6323*}                                               & $\backslash$ & $\backslash$ & $\backslash$                                        & $\backslash$                                          & $\backslash$ & 0.6169                                                  & -2.44\% \\ \hline
\multirow{4}{*}{\begin{tabular}[c]{@{}c@{}}Book-\\ Crossing\end{tabular}} & Pre@5   & 0.0145 & 0.0109 & 0.0127 & 0.0091           & 0.0136                                               & $\backslash$ & $\backslash$ & $\backslash$                                        & $\backslash$                                          & $\backslash$ & \textbf{0.0184*}                                        & 26.40\% \\ \cline{2-14}
                                                                          & Rec@5   & 0.0096 & 0.0096 & 0.0108 & 0.0058           & 0.0129                                               & $\backslash$ & $\backslash$ & $\backslash$                                        & $\backslash$                                          & $\backslash$ & \textbf{0.0177*}                                        & 37.30\% \\ \cline{2-14}
                                                                          & NDCG    & 0.1683 & 0.1712 & 0.1705 & 0.1506           & 0.1790                                               & $\backslash$ & $\backslash$ & $\backslash$                                        & $\backslash$                                          & $\backslash$ & \textbf{0.1829*}                                        & 2.15\%  \\ \cline{2-14}
                                                                          & MRR     & 0.0433 & 0.0367 & 0.0400 & 0.0308           & 0.0455                                               & $\backslash$ & $\backslash$ & $\backslash$                                        & $\backslash$                                          & $\backslash$ & \textbf{0.0571*}                                        & 25.36\% \\ \hline
\end{tabular}
 \vspace{-0.3cm}
\end{table*}

\subsection{Performance comparison (Q1)(Q2)}
 Table \ref{tb:re} presents the performance of the compared methods in terms of four evaluation metrics. The mark '*' denotes the winner in that row, while the boldface font denotes the winner among the non-social recommendation methods (WMF(ALS), BPR, EXMF, FAWMF, LightGCN, SamWalker++). Overall, except the results in the dataset Movielens, SamWalker or SamWalker++ outperform all compared baselines on all datasets for all metrics. For the sake of clarity, the columns 'impv1' and 'impv2' also show the relative improvement achieved by SamWalker over the all baselines and SamWalker++ over the non-social baselines respectively. The improvements are quite impressive.

 \textbf{Effect of modeling user's exposure.} In the real world, users usually have personalized social contexts and thus are exposed to diverse information. The exposure-based methods, which is capable of adaptively learning fine-grained data confidence weights, usually achieve better performance than the methods with manually assigned confidence weights. It can be seen from the experimental results that the best results are always achieved by the exposure-based methods.

 \textbf{Comparing with exposure-based methods.} Generally, our proposed SamWalker and SamWalker++ achieve better performance than existing exposure-based methods. The superiority can be attributed to two reasons: (1) The vanilla exposure-based method will easily suffer from over-fitting problem which deteriorates the recommendation accuracy, while our methods can mitigates the over-fitting effect by leveraging (pseudo-)social network. (2) Our methods employ an informative sampler which has low sampling variance, making the performance of stochastic learning is comparable with the full-batch learning. We remark that adopting full-batch learning is non-optimal, it either suffers from low efficiency (e.g., EXMF, SERec-Bo and SoEXBMF) or sacrifices models' flexibility with memorization mechanism (e.g., FAWMF). Although FAWMF achieves good performance in the dense dataset Movielens, it performs quite poorly in the sparse datasets such as BookCrossing and Epinions, even worse than the simple baselines.

 \textbf{Comparing in terms of datasets.} We can find that SamWalker or (SamWalker++) outperforms all compared methods in all datasets except Movielens. This interesting phenomenon is caused by the sparsity of user-item interactions and can be explained as follows: (1)Movielens is a quite dense dataset and contains sufficient user-item interactions to supervise the learning of user exposure. As a result, directly optimizing loss function on such rich data would yield pretty good performance, while further leveraging pseudo-social network into the learning does not make much progress. (2) SamWalker++ uses the stochastic learning strategy, which may sacrifice a certain accuracy comparing with full-batch-based learning strategy, making FAWMF even perform slightly better than SamWalker++ in the dataset Movielens.

\textbf{SamWalker++ Vs. SamWalker.} Although SamWalker++ does not use social information, SamWalker++ still achieves comparable performance with SamWalker. Especially in the dataset Epinions, where the social information is not as abundant as LastFM, SamWalker++ even outperforms SamWalker. These results validate the effectiveness of the constructed pseudo-social network, which encodes rich similarity information between users as the real social network.

\begin{figure}[t]
  \centering
  \includegraphics[width=0.49\textwidth]{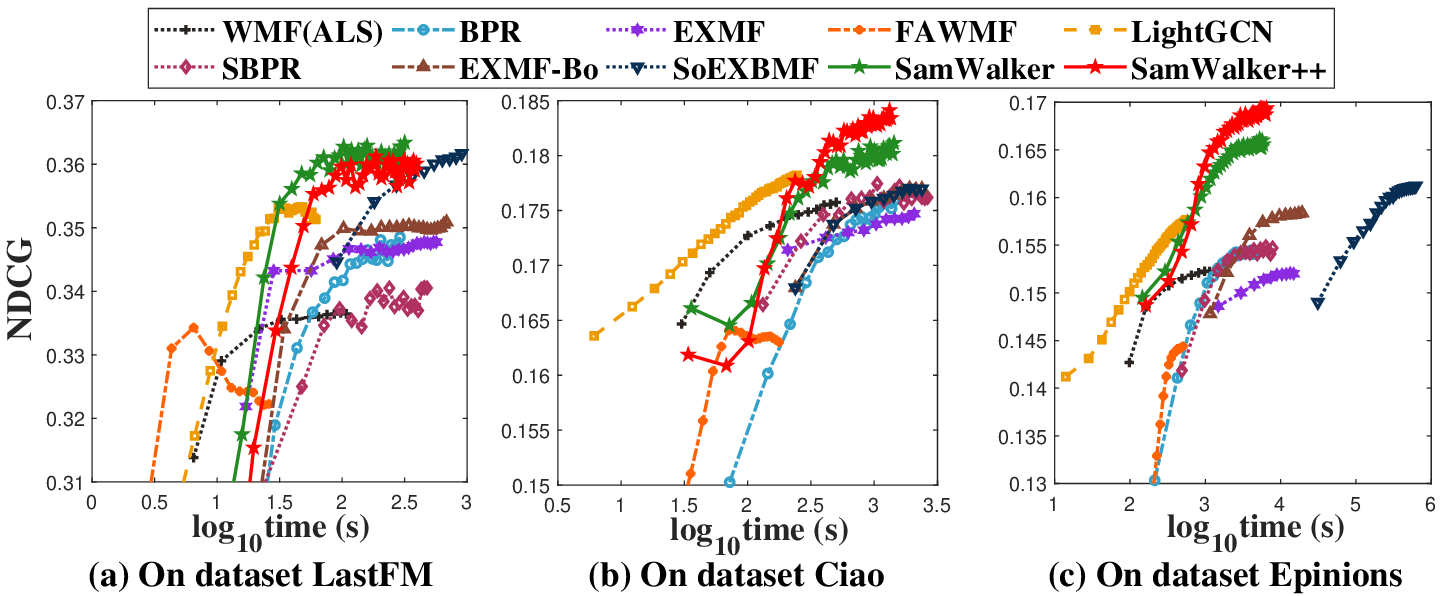}\\
  \caption{NDCG for each method in different steps versus time.}\label{fg:ti}
   \vspace{-0.3cm}
\end{figure}

\begin{table*}[t!]
\centering
\scriptsize
\caption{Average empirical variance of the estimated gradients using different sampling strategies. Here $X^{(1)}$, $X^{(0)}$ denote the number of ones or zeros in the matrix $X$. Similarly, $r^{(1)}_i$, $r^{(0)}_i$ denote the number of ones or zeros in the $i$-th row of matrix $X$ and $c^{(1)}_j$, $c^{(0)}_j$ denote the number of ones or zeros in the $j$-th column of matrix $X$.}
\label{tb:var}
\begin{tabular}{|c|c|c|c|c|c|c|c|}
\hline
\multirow{3}{*}{\begin{tabular}[c]{@{}c@{}}Sampling\\ strategy\end{tabular}} & \multirow{3}{*}{Distribution}                                                                                                                                                         & \multicolumn{6}{c|}{Average Variance}                                     \\ \cline{3-8}
                                                                             &                                                                                                                                                                                       & \multicolumn{3}{c|}{For SamWalker} & \multicolumn{3}{c|}{For SamWalker++} \\ \cline{3-8}
                                                                             &                                                                                                                                                                                       & 50 It.    & 100 It.    & 500 It.   & 50 It.     & 100 It.    & 500 It.    \\ \hline
S-allunion                                                                   & ${p_{ui}} = 1/(n \times m)$                                                                                                                                                           & 2.3793    & 2.5640     & 2.6663    & 2.6411     & 2.8059     & 2.8003     \\ \hline
S-balunion                                                                   & ${p_{ui}} = 1/(2|{X^{({x_{ui}})}}|)$                                                                                                                                                  & 0.4154    & 0.4941     & 0.5368    & 0.3691     & 0.4381     & 0.4639     \\ \hline
S-itempop                                                                    & ${p_{ui}} = I[{x_{ui}} = 1]/(2|{X^{(1)}}|) + I[{x_{ui}} = 0]c_i^{(1)}/(2\sum\limits_{1 \le b \le m} {c_b^{(1)}} )$                                                                    & 0.1720    & 0.1938     & 0.2190    & 0.1382     & 0.1537     & 0.1583     \\ \hline
S-cobias                                                                     & ${p_{ui}} = r_u^{(1 - {x_{ui}})}c_i^{(1 - {x_{ui}})}/(2\sum\limits_{1 \le a \le n} {\sum\limits_{1 \le b \le m} {I[{x_{ab}} = {x_{ui}}]r_a^{(1 - {x_{ui}})}c_b^{(1 - {x_{ui}})}} } )$ & 0.1617    & 0.1874     & 0.2031    & 0.1562     & 0.1762     & 0.1859     \\ \hline
Random Walk                                                                  & $p_{ui}\propto \gamma_{ui} $                                                                                                                                                          & 0.1358    & 0.1464     & 0.1510    & 0.1296     & 0.1395     & 0.1394     \\ \hline
\end{tabular}
\vspace{-0.2cm}
\end{table*}

\textbf{Runtime vs. NDCG.} Figure \ref{fg:ti} depicts running time vs. NDCG of the compared methods. Generally, SamWalker or SamWalker++ achieve best performance. The powerful competitor is LightGCN. Although LightGCN has better NDCG than SamWalker++ at the beginning, SamWalker++ overtakes LightGCN soon with few iterations and achieves much better performance than LightGCN finally. Also, we observe that these exposure-based methods (EXMF,SERec,SoEXBMF) achieve good performance but are computational expensive. FAWMF is quite efficient, but its performance is quite poor in these sparse datasets.

\subsection{Sampler comparisons (Q3)}
In this subsection, we conduct two types of sampler comparisons: (1) To empirically validate the correctness of the Lemmas 1 and the effectiveness of our random-walker-based sampler, we first test SamWalker (or SamWalker++) using different sampling strategies. Note that the sampling strategies determine the frequency of the data used for updating the model and may skew the data contribution. For fair comparison, we offset the sampling bias and make the compared methods yield unbiased gradients. (2) We conduct a direct comparison of SamWalker (SamWalker++) with existing samplers, where we do not offset sampling bias and the data contribution is affected by the sampling strategy.

\textbf{SamWalker (or SamWalker++) using different sampling strategies.} We compare our random-walk-based sampler with other sampling strategies including: (1) S-allunion, the global uniform sampling strategy; (2) S-balunion \cite{Pan2008,hernandez2014stochastic}, which samples un-observed data (zeros) and observed data (ones) with equal probability to deal with unbalance data problem; (3) S-itempop \cite{yu2017selection}, which samples instances based on item popularity; (4) S-cobias \cite{hernandez2014stochastic}, whose sampling distribution is proportional to user/item popularity. The detailed distributions of these sampling strategies are presented in Table \ref{tb:var}. To offset the bias introduced by the sampling strategies, we will weight the data with the inverse of the sampling probability. We remark that here we do not consider some sophisticated samplers \cite{wang2017irgan,ding2019reinforced} for comparison, as their sampling distribution is hard to estimate. Also, their debiased gradients are usually instable and potentially exploded.

We first empirically compare the variance of the estimated gradients of our SamWalker or SamWalker++ using different sampling strategies. To do this, we train our models for 50, 100 or 500 epochs on the dataset LastFM. After training, we generate mini-batch with various sampling strategies and calculate un-biased estimated gradients of our objective w.r.t $\theta$. We repeat this precedure for 1000 times and calculate the variance of the estimated gradients for different sampling strategies. The final results presented in Table \ref{tb:var} are averaged over $\theta$. Our random walk-based sampling strategy achieves the lowest average variance among various samplers for all conditions. This result is coincident with Lemma 1. Also, we observe the following interesting phenomenon: With more training epoches, the variance will become larger, not smaller as usual. It may be explained as follow:  SamWalker and SamWalker++ initialize with relative similar edge weights. With training proceeding, driven by the data, the edge weights and data confidence exhibit more and more heterogeneity. Thus, the variance of the estimated gradients will become larger.

We also presents the NDCG of SamWalker on the dataset LastFM with different sampling strategies versus the number of iterations and running time in Figure \ref{fg:sam}. As we can see, our random walk-based sampler performs better than others in all convergence, speed and accuracy. One reason is that our sampler has low variance (Lemma 1). Another reason is that in our sampler the confidence weights $\gamma_{ui}$ are absorbed into the sampling distribution and does not need calculating, which saves much time (Lemma 2).

\begin{figure}
  \centering
  \includegraphics[width=0.46\textwidth]{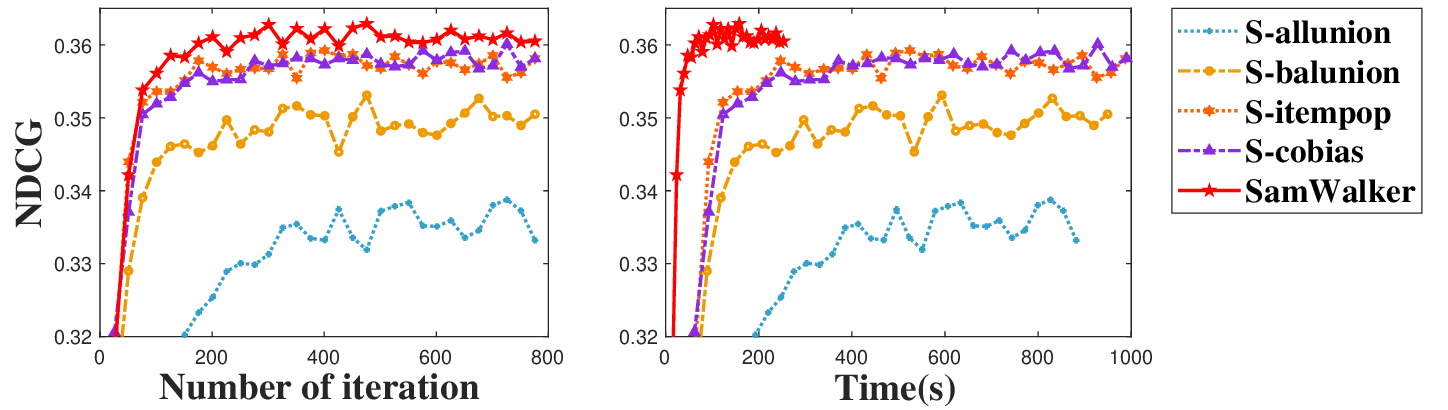}\\
  \caption{NDCG for different sampling strategy versus the number of iterations and running time.}\label{fg:sam}
   \vspace{-0.3cm}
\end{figure}

\begin{table}[t!]
\centering
\scriptsize
\caption{Comparison of SamWalker and SamWalker++ with state-of-the-arts.}
\label{tb:com}
\begin{tabular}{|c|c|c|c|c|c|c|}
\hline
\multirow{2}{*}{Methods} & \multicolumn{2}{c|}{LastFM} & \multicolumn{2}{c|}{Ciao} & \multicolumn{2}{c|}{Epinions} \\ \cline{2-7}
                         & Pre@5        & NDCG         & Pre@5       & NDCG        & Pre@5         & NDCG          \\ \hline
Adaptive                     & 0.0986       & 0.3448       & 0.0146      & 0.1784      & 0.0088        & 0.1559        \\ \hline
Adversarial                 & 0.1002       & 0.3503       & 0.0153      & 0.1786      & 0.0094        & 0.1567        \\ \hline
SamWalker                & 0.1177       & 0.3634       & 0.0182      & 0.1811      & 0.0149        & 0.1656        \\ \hline
SamWalker++              & 0.1099       & 0.3601       & 0.0184      & 0.1834      & 0.0165        & 0.1693        \\ \hline
\end{tabular}
\vspace{-0.3cm}
\end{table}

\textbf{Comparing with state-of-the-art samplers.} We also conduct a direct comparison of our SamWalker (SamWalker++) with existing samplers where the sampling bias has not been removed. We choose two state-of-the-art sampling strategies for comparison: (1) the adaptive hard sampler, which over-samples the ``difficult'' negative instances that yield large gradients \cite{rendle2014improving}; (2) the adversarial sampler, which employs adversarial learning to find the instances that are difficult to be differentiated by the recommendation model. Since the implementations of \cite{ding2019reinforced,wang2020reinforced} require other side information, we follow \cite{wang2017irgan,park2019adversarial} and implement the sampler model with the matrix factorisation. From Table \ref{tb:com}, we can find our SamWalker and SamWalker++ consistently outperform the two strong baselines. This result can be attributed to the fact that the ``difficult'' instances do not suggest the instances are informative enough deserving over-training. These instances may be just caused by the non-exposure and contain limited information of user preference. Over-sampling such uninformative instances would deteriorate model performance.

\subsection{Ablation study (Q4)}
We remove different components at a time and compare SamWalker++ with its two special cases: (1) SamWalker++noc, the special case of SamWalker++ which leaves out community nodes; (2) SamWalker++noi: the special case without item nodes. The performance is presented in Figure \ref{fg:youvswu}. We observe that SamWalker++ consistently outperforms its two special cases. This result validates the effectiveness of introducing both item nodes and community nodes to transfer the knowledge.
\begin{figure*}[t]
  \centering
  \subfigure[On  LastFM]{
\includegraphics[width=0.16\textwidth]{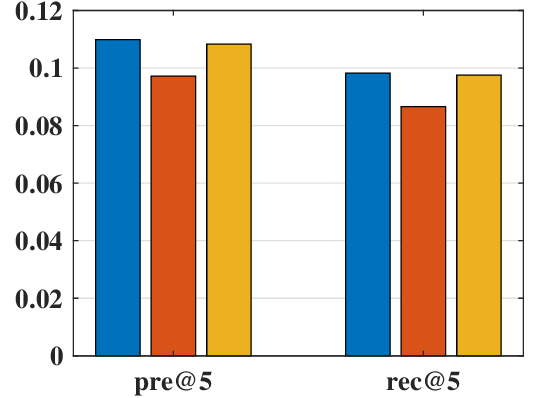}
}
\subfigure[On Ciao]{
\includegraphics[width=0.16\textwidth]{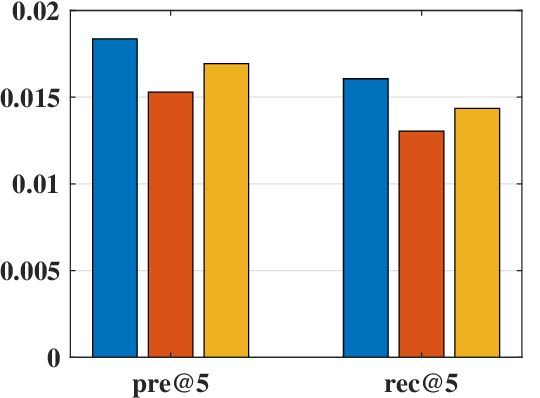}
}
\subfigure[On  Epinions]{
\includegraphics[width=0.16\textwidth]{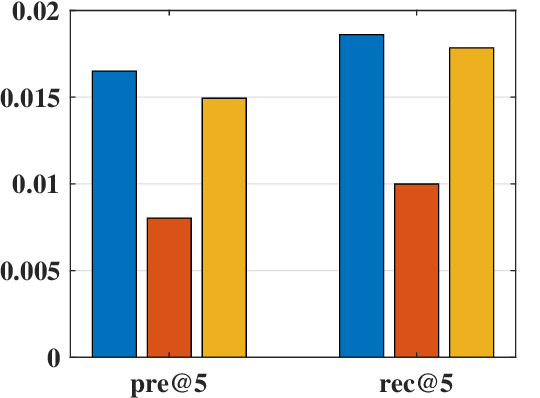}
}
\subfigure[On Moivelens]{
\includegraphics[width=0.16\textwidth]{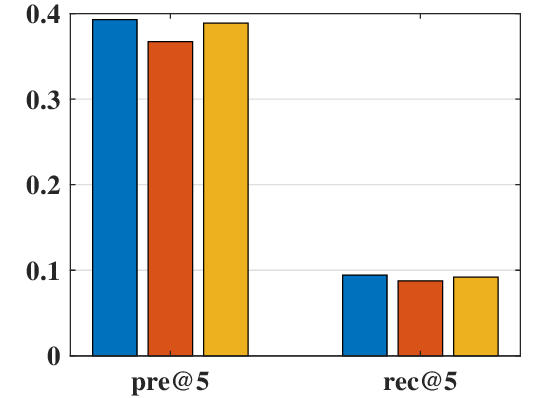}
}
\subfigure[On BookCrossing]{
\includegraphics[width=0.16\textwidth]{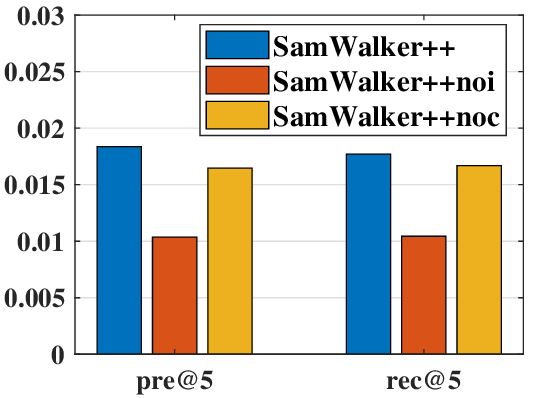}
}
  \caption{Performance comparison of SamWalker++ with its two special cases in terms of Pre@5 and Rec@5. }\label{fg:youvswu}
  \vspace{-0.3cm}
\end{figure*}

\begin{figure*}[t]
  \centering
\subfigure[On  LastFM]{
\includegraphics[width=0.16\textwidth]{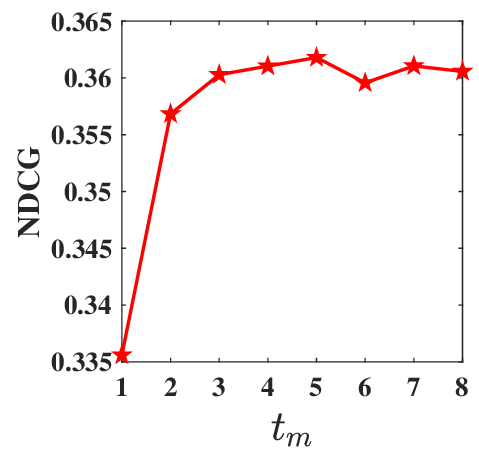}
}
\subfigure[On Ciao]{
\includegraphics[width=0.16\textwidth]{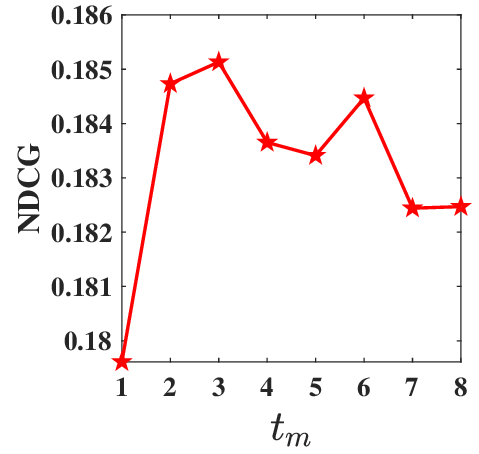}
}
\subfigure[On Epinions]{
\includegraphics[width=0.16\textwidth]{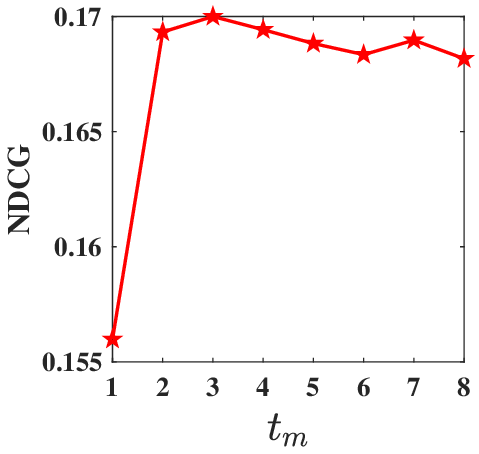}
}
\subfigure[On Moivelens]{
\includegraphics[width=0.16\textwidth]{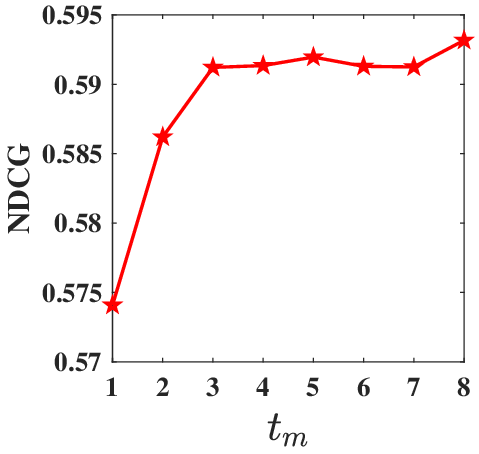}
}
\subfigure[On BookCrossing]{
\includegraphics[width=0.16\textwidth]{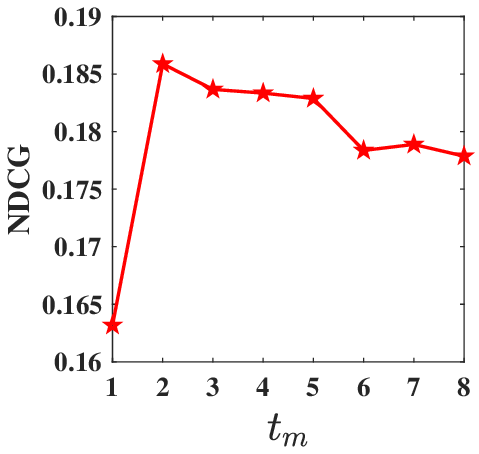}
}
\caption{Performance comparison of SamWalker++ with varying $t_m$. }\label{fg:ndvsdep}
\vspace{-0.3cm}
\end{figure*}

\subsection{Effect of parameter $t_m$ (Q5)}

Figure \ref{fg:ndvsdep} shows how parameter $t_m$ affects the performance of SamWalker++, where $t_m$ indicates the max depth of the random walk.  With $t_m$ increasing, with few exception, the performance will become better at the beginning. This result validates the effectiveness of leveraging network to transfer the knowledge. But when $t_m$ surpasses a threshold, the performance becomes unaffected or even experiences some degradation with further increase of $t_m$. Too deep random walk will bring more irrelevant factors and even data noises, which deteriorates recommendation accuracy.

\section{Conclusion}

Data confidence assignment and efficient model learning are two key problems in the implicit recommendation task. In this paper, we present two novel recommendation methods SamWalker and SamWalker++ to address these problems. SamWalker models the data confidence with a social context-aware function, which can reduce the number of learned parameters and adaptively specify personalized confidence weights for implicit feedback data. As the real social network data is not easily obtained, we instead construct pseudo-social network where similar users are connected with specific item nodes or community nodes. SamWalker++ models data confidence on the pseudo-social network, so that the inference of a user's exposure can benefit from other similar users. We further propose a random walker-based sampling strategy to draw informative training instances to speed up inference and reduce sampling variance. Extensive experimental results on five real-world datasets demonstrate the superiority of SamWalker and SamWalker++ over existing methods.

One interesting direction for future work is to leverage sophisticated graph neural network \cite{dong2020equivalence} in the exposure model, which is capable of capturing more complex affinity between users and items along the interaction graph. Also, it will be interesting to explore dynamic exposure-based recommendation, as in the real world, users' preference, exposure and relations may evolve over time. Note that knowledge graph captures much more rich information, which could be useful to capture user exposure. It is promising that exploiting the knowledge graph to improve the efficacy of SamWalker++, especially for the cold users or items that have limited interactions.

\appendices

\ifCLASSOPTIONcompsoc
  \section*{Acknowledgments}
\else
  \section*{Acknowledgment}
\fi

This work is supported by National Key R\&D Program of China (Grant No:  2019YFB1600700, 2018AAA0101505) and National Natural Science Foundation of China (Grant No: U1866602).

\ifCLASSOPTIONcaptionsoff
  \newpage
\fi

\bibliographystyle{IEEEtran}
\bibliography{sigproc}

\begin{thebibliography}{10}
\providecommand{\url}[1]{#1}
\csname url@samestyle\endcsname
\providecommand{\newblock}{\relax}
\providecommand{\bibinfo}[2]{#2}
\providecommand{\BIBentrySTDinterwordspacing}{\spaceskip=0pt\relax}
\providecommand{\BIBentryALTinterwordstretchfactor}{4}
\providecommand{\BIBentryALTinterwordspacing}{\spaceskip=\fontdimen2\font plus
\BIBentryALTinterwordstretchfactor\fontdimen3\font minus
  \fontdimen4\font\relax}
\providecommand{\BIBforeignlanguage}[2]{{%
\expandafter\ifx\csname l@#1\endcsname\relax
\typeout{** WARNING: IEEEtran.bst: No hyphenation pattern has been}%
\typeout{** loaded for the language `#1'. Using the pattern for}%
\typeout{** the default language instead.}%
\else
\language=\csname l@#1\endcsname
\fi
#2}}
\providecommand{\BIBdecl}{\relax}
\BIBdecl

\bibitem{ricci2015recommender}
F.~Ricci, L.~Rokach, B.~Shapira, and P.~B. Kantor, \emph{Recommender systems
  handbook}.\hskip 1em plus 0.5em minus 0.4em\relax Springer, 2015.

\bibitem{jannach2010recommender}
D.~Jannach, M.~Zanker, A.~Felfernig, and G.~Friedrich, \emph{Recommender
  systems: an introduction}.\hskip 1em plus 0.5em minus 0.4em\relax Cambridge
  University Press, 2010.

\bibitem{hu2008collaborative}
Y.~Hu, Y.~Koren, and C.~Volinsky, ``Collaborative filtering for implicit
  feedback datasets,'' in \emph{Data Mining, 2008. ICDM'08. Eighth IEEE
  International Conference on}.\hskip 1em plus 0.5em minus 0.4em\relax Ieee,
  2008, pp. 263--272.

\bibitem{he2016fast}
X.~He, H.~Zhang, M.-Y. Kan, and T.-S. Chua, ``Fast matrix factorization for
  online recommendation with implicit feedback,'' in \emph{Proceedings of the
  39th International ACM SIGIR conference on Research and Development in
  Information Retrieval}.\hskip 1em plus 0.5em minus 0.4em\relax ACM, 2016, pp.
  549--558.

\bibitem{chen2020fast}
J.~Chen, C.~Wang, S.~Zhou, Q.~Shi, J.~Chen, Y.~Feng, and C.~Chen, ``Fast
  adaptively weighted matrix factorization for recommendation with implicit
  feedback.'' in \emph{AAAI}, 2020, pp. 3470--3477.

\bibitem{he2017neural}
X.~He, L.~Liao, H.~Zhang, L.~Nie, X.~Hu, and T.-S. Chua, ``Neural collaborative
  filtering,'' in \emph{Proceedings of the 26th International Conference on
  World Wide Web}.\hskip 1em plus 0.5em minus 0.4em\relax ACM, 2017, pp.
  173--182.

\bibitem{chen2018modeling}
J.~Chen, Y.~Feng, M.~Ester, S.~Zhou, C.~Chen, and C.~Wang, ``Modeling users'
  exposure with social knowledge influence and consumption influence for
  recommendation,'' in \emph{Proceedings of the 27th ACM International on
  Conference on Information and Knowledge Management}.\hskip 1em plus 0.5em
  minus 0.4em\relax ACM, 2018, pp. 953--962.

\bibitem{lichman2018prediction}
M.~Lichman and P.~Smyth, ``Prediction of sparse user-item consumption rates
  with zero-inflated poisson regression,'' in \emph{The World Wide Web
  Conference}.\hskip 1em plus 0.5em minus 0.4em\relax IW3C2, 2018, pp.
  719--728.

\bibitem{nie2016learning}
L.~Nie, X.~Song, and T.-S. Chua, ``Learning from multiple social networks,''
  \emph{Synthesis lectures on information concepts, retrieval, and services},
  vol.~8, no.~2, pp. 1--118, 2016.

\bibitem{Pan2017}
X.~Pan, L.~Hou, and K.~Liu, ``Social influence on selection behaviour:
  Distinguishing local-and global-driven preferential attachment,'' \emph{PloS
  one}, vol.~12, no.~4, p. e0175761, 2017.

\bibitem{chen2013information}
W.~Chen, L.~V. Lakshmanan, and C.~Castillo, ``Information and influence
  propagation in social networks,'' \emph{Synthesis Lectures on Data
  Management}, vol.~5, no.~4, pp. 1--177, 2013.

\bibitem{liang2016modeling}
D.~Liang, L.~Charlin, J.~McInerney, and D.~M. Blei, ``Modeling user exposure in
  recommendation,'' in \emph{Proceedings of the 25th International Conference
  on World Wide Web}.\hskip 1em plus 0.5em minus 0.4em\relax ACM, 2016, pp.
  951--961.

\bibitem{palla2005uncovering}
G.~Palla, I.~Der{\'e}nyi, I.~Farkas, and T.~Vicsek, ``Uncovering the
  overlapping community structure of complex networks in nature and society,''
  \emph{Nature}, vol. 435, no. 7043, pp. 814--818, 2005.

\bibitem{zhou2011understanding}
T.~Zhou, ``Understanding online community user participation: a social
  influence perspective,'' \emph{Internet research}, vol.~21, no.~1, pp.
  67--81, 2011.

\bibitem{chen2019samwalker}
J.~Chen, C.~Wang, S.~Zhou, Q.~Shi, Y.~Feng, and C.~Chen, ``Samwalker: Social
  recommendation with informative sampling strategy,'' in \emph{The World Wide
  Web Conference}.\hskip 1em plus 0.5em minus 0.4em\relax ACM, 2019, pp.
  228--239.

\bibitem{wu2016collaborative}
Y.~Wu, C.~DuBois, A.~X. Zheng, and M.~Ester, ``Collaborative denoising
  auto-encoders for top-n recommender systems,'' in \emph{Proceedings of the
  Ninth ACM International Conference on Web Search and Data Mining}.\hskip 1em
  plus 0.5em minus 0.4em\relax ACM, 2016, pp. 153--162.

\bibitem{he2020lightgcn}
X.~He, K.~Deng, X.~Wang, Y.~Li, Y.~Zhang, and M.~Wang, ``Lightgcn: Simplifying
  and powering graph convolution network for recommendation,'' \emph{arXiv
  preprint arXiv:2002.02126}, 2020.

\bibitem{yu2017selection}
H.-F. Yu, M.~Bilenko, and C.-J. Lin, ``Selection of negative samples for
  one-class matrix factorization,'' in \emph{Proceedings of the 2017 SIAM
  International Conference on Data Mining}.\hskip 1em plus 0.5em minus
  0.4em\relax SIAM, 2017, pp. 363--371.

\bibitem{chen2020bias}
J.~Chen, H.~Dong, X.~Wang, F.~Feng, M.~Wang, and X.~He, ``Bias and debias in
  recommender system: A survey and future directions,'' \emph{arXiv preprint
  arXiv:2010.03240}, 2020.

\bibitem{rendle2009bpr}
S.~Rendle, C.~Freudenthaler, Z.~Gantner, and L.~Schmidt-Thieme, ``Bpr: Bayesian
  personalized ranking from implicit feedback,'' in \emph{Proceedings of the
  twenty-fifth conference on uncertainty in artificial intelligence}.\hskip 1em
  plus 0.5em minus 0.4em\relax AUAI Press, 2009, pp. 452--461.

\bibitem{hernandez2014stochastic}
J.~M. Hern{\'a}ndez-Lobato, N.~Houlsby, and Z.~Ghahramani, ``Stochastic
  inference for scalable probabilistic modeling of binary matrices,'' in
  \emph{International Conference on Machine Learning}, 2014, pp. 379--387.

\bibitem{chen2017sampling}
T.~Chen, Y.~Sun, Y.~Shi, and L.~Hong, ``On sampling strategies for neural
  network-based collaborative filtering,'' in \emph{Proceedings of the 23rd ACM
  SIGKDD International Conference on Knowledge Discovery and Data
  Mining}.\hskip 1em plus 0.5em minus 0.4em\relax ACM, 2017, pp. 767--776.

\bibitem{yu2018walkranker}
L.~Yu, C.~Zhang, S.~Pei, G.~Sun, and X.~Zhang, ``Walkranker: A unified pairwise
  ranking model with multiple relations for item recommendation,'' \emph{AAAI},
  2018.

\bibitem{rendle2014improving}
S.~Rendle and C.~Freudenthaler, ``Improving pairwise learning for item
  recommendation from implicit feedback,'' in \emph{Proceedings of the 7th ACM
  international conference on Web search and data mining}.\hskip 1em plus 0.5em
  minus 0.4em\relax ACM, 2014, pp. 273--282.

\bibitem{zhang2013optimizing}
W.~Zhang, T.~Chen, J.~Wang, and Y.~Yu, ``Optimizing top-n collaborative
  filtering via dynamic negative item sampling,'' in \emph{Proceedings of the
  36th international ACM SIGIR conference on Research and development in
  information retrieval}.\hskip 1em plus 0.5em minus 0.4em\relax ACM, 2013, pp.
  785--788.

\bibitem{DBLP:conf/www/WangX000C20}
X.~Wang, Y.~Xu, X.~He, Y.~Cao, M.~Wang, and T.~Chua, ``Reinforced negative
  sampling over knowledge graph for recommendation,'' in \emph{{WWW} '20: The
  Web Conference 2020, Taipei, Taiwan, April 20-24, 2020}.\hskip 1em plus 0.5em
  minus 0.4em\relax {ACM} / {IW3C2}, 2020, pp. 99--109.

\bibitem{li2018adaerror}
D.~Li, C.~Chen, Q.~Lv, H.~Gu, T.~Lu, L.~Shang, N.~Gu, and S.~M. Chu,
  ``Adaerror: An adaptive learning rate method for matrix approximation-based
  collaborative filtering,'' in \emph{Proceedings of the 2018 World Wide Web
  Conference on World Wide Web}.\hskip 1em plus 0.5em minus 0.4em\relax
  International World Wide Web Conferences Steering Committee, 2018, pp.
  741--751.

\bibitem{ding2018improved}
J.~Ding, F.~Feng, X.~He, G.~Yu, Y.~Li, and D.~Jin, ``An improved sampler for
  bayesian personalized ranking by leveraging view data,'' in \emph{Companion
  of the The Web Conference 2018 on The Web Conference 2018}.\hskip 1em plus
  0.5em minus 0.4em\relax IW3C2, 2018, pp. 13--14.

\bibitem{ding2019reinforced}
J.~Ding, Y.~Quan, X.~He, Y.~Li, and D.~Jin, ``Reinforced negative sampling for
  recommendation with exposure data,'' in \emph{Proceedings of the 28th
  International Joint Conference on Artificial Intelligence}.\hskip 1em plus
  0.5em minus 0.4em\relax AAAI Press, 2019, pp. 2230--2236.

\bibitem{bayer2017generic}
I.~Bayer, X.~He, B.~Kanagal, and S.~Rendle, ``A generic coordinate descent
  framework for learning from implicit feedback,'' in \emph{Proceedings of the
  26th International Conference on World Wide Web}.\hskip 1em plus 0.5em minus
  0.4em\relax IW3C2, 2017, pp. 1341--1350.

\bibitem{golbeck2009trust}
J.~Golbeck, ``Trust and nuanced profile similarity in online social networks,''
  \emph{ACM Transactions on the Web (TWEB)}, vol.~3, no.~4, p.~12, 2009.

\bibitem{ma2008sorec}
H.~Ma, H.~Yang, M.~R. Lyu, and I.~King, ``Sorec: social recommendation using
  probabilistic matrix factorization,'' in \emph{CIKM}.\hskip 1em plus 0.5em
  minus 0.4em\relax ACM, 2008, pp. 931--940.

\bibitem{yang2013social}
B.~Yang, Y.~Lei, D.~Liu, and J.~Liu, ``Social collaborative filtering by
  trust,'' in \emph{Proceedings of the Twenty-Third international joint
  conference on Artificial Intelligence}.\hskip 1em plus 0.5em minus
  0.4em\relax AAAI Press, 2013, pp. 2747--2753.

\bibitem{shen2012learning}
Y.~Shen and R.~Jin, ``Learning personal+ social latent factor model for social
  recommendation,'' in \emph{Proceedings of the 18th ACM SIGKDD international
  conference on Knowledge discovery and data mining}.\hskip 1em plus 0.5em
  minus 0.4em\relax ACM, 2012, pp. 1303--1311.

\bibitem{chaney2015probabilistic}
A.~J. Chaney, D.~M. Blei, and T.~Eliassi-Rad, ``A probabilistic model for using
  social networks in personalized item recommendation,'' in \emph{Proceedings
  of the 9th ACM Conference on Recommender Systems}.\hskip 1em plus 0.5em minus
  0.4em\relax ACM, 2015, pp. 43--50.

\bibitem{xiao2017learning}
L.~Xiao, Z.~Min, Z.~Yongfeng, L.~Yiqun, and M.~Shaoping, ``Learning and
  transferring social and item visibilities for personalized recommendation,''
  in \emph{Proceedings of the 2017 ACM on Conference on Information and
  Knowledge Management}.\hskip 1em plus 0.5em minus 0.4em\relax ACM, 2017, pp.
  337--346.

\bibitem{jamali2010matrix}
M.~Jamali and M.~Ester, ``A matrix factorization technique with trust
  propagation for recommendation in social networks,'' in \emph{Proceedings of
  the fourth ACM conference on Recommender systems}.\hskip 1em plus 0.5em minus
  0.4em\relax ACM, 2010, pp. 135--142.

\bibitem{wang2017learning}
X.~Wang, S.~C. Hoi, M.~Ester, J.~Bu, and C.~Chen, ``Learning personalized
  preference of strong and weak ties for social recommendation,'' in
  \emph{Proceedings of the 26th International Conference on World Wide
  Web}.\hskip 1em plus 0.5em minus 0.4em\relax IW3C2, 2017, pp. 1601--1610.

\bibitem{wang2016social}
X.~Wang, W.~Lu, M.~Ester, C.~Wang, and C.~Chen, ``Social recommendation with
  strong and weak ties,'' in \emph{Proceedings of the 25th ACM International on
  Conference on Information and Knowledge Management}.\hskip 1em plus 0.5em
  minus 0.4em\relax ACM, 2016, pp. 5--14.

\bibitem{zhao2014leveraging}
T.~Zhao, J.~McAuley, and I.~King, ``Leveraging social connections to improve
  personalized ranking for collaborative filtering,'' in \emph{Proceedings of
  the 23rd ACM International Conference on Conference on Information and
  Knowledge Management}.\hskip 1em plus 0.5em minus 0.4em\relax ACM, 2014, pp.
  261--270.

\bibitem{jiawei2018social}
J.~Chen, C.~Wang, M.~Ester, Q.~Shi, Y.~Feng, and C.~Chen, ``Social
  recommendation with missing not at random data,'' in \emph{2018 IEEE
  International Conference on Data Mining (ICDM)}.\hskip 1em plus 0.5em minus
  0.4em\relax IEEE, 2018, pp. 29--38.

\bibitem{DBLP:conf/aaai/WangZYZ18}
M.~Wang, X.~Zheng, Y.~Yang, and K.~Zhang, ``Collaborative filtering with social
  exposure: {A} modular approach to social recommendation,'' in \emph{AAAI, New
  Orleans, Louisiana, USA, February 2-7, 2018}, 2018.

\bibitem{jamali2009trustwalker}
M.~Jamali and M.~Ester, ``Trustwalker: a random walk model for combining
  trust-based and item-based recommendation,'' in \emph{Proceedings of the 15th
  ACM SIGKDD international conference on Knowledge discovery and data
  mining}.\hskip 1em plus 0.5em minus 0.4em\relax ACM, 2009, pp. 397--406.

\bibitem{christoffel2015blockbusters}
F.~Christoffel, B.~Paudel, C.~Newell, and A.~Bernstein, ``Blockbusters and
  wallflowers: Accurate, diverse, and scalable recommendations with random
  walks,'' \emph{Conference on Recommender Systems}, 2015.

\bibitem{vahedian2017weighted}
F.~Vahedian, D.~R. Burke, and B.~Mobasher, ``Weighted random walk sampling for
  multi-relational recommendation,'' \emph{UMAP}, 2017.

\bibitem{rendle2012factorization}
S.~Rendle, ``Factorization machines with libfm,'' \emph{ACM Transactions on
  Intelligent Systems and Technology (TIST)}, vol.~3, no.~3, pp. 1--22, 2012.

\bibitem{hoffman2013stochastic}
M.~D. Hoffman, D.~M. Blei, C.~Wang, and J.~Paisley, ``Stochastic variational
  inference,'' \emph{The Journal of Machine Learning Research}, vol.~14, no.~1,
  pp. 1303--1347, 2013.

\bibitem{page1999pagerank}
L.~Page, S.~Brin, R.~Motwani, and T.~Winograd, ``The pagerank citation ranking:
  Bringing order to the web.'' Stanford InfoLab, Tech. Rep., 1999.

\bibitem{zhou2004learning}
D.~Zhou, O.~Bousquet, T.~N. Lal, J.~Weston, and B.~Sch{\"o}lkopf, ``Learning
  with local and global consistency,'' in \emph{Advances in neural information
  processing systems}, 2004, pp. 321--328.

\bibitem{canamares2018should}
R.~Ca{\~n}amares and P.~Castells, ``Should i follow the crowd?: A probabilistic
  analysis of the effectiveness of popularity in recommender systems,'' in
  \emph{The 41st International ACM SIGIR Conference on Research \& Development
  in Information Retrieval}.\hskip 1em plus 0.5em minus 0.4em\relax ACM, 2018,
  pp. 415--424.

\bibitem{shi2012climf}
Y.~Shi, A.~Karatzoglou, L.~Baltrunas, M.~Larson, N.~Oliver, and A.~Hanjalic,
  ``Climf: learning to maximize reciprocal rank with collaborative less-is-more
  filtering,'' in \emph{Proceedings of the sixth ACM conference on Recommender
  systems}, 2012, pp. 139--146.

\bibitem{Pan2008}
R.~Pan, Y.~Zhou, B.~Cao, N.~N. Liu, R.~Lukose, M.~Scholz, and Q.~Yang,
  ``{One-class collaborative filtering},'' \emph{ICDM}, pp. 502--511, 2008.

\bibitem{wang2017irgan}
J.~Wang, L.~Yu, W.~Zhang, Y.~Gong, Y.~Xu, B.~Wang, P.~Zhang, and D.~Zhang,
  ``Irgan: A minimax game for unifying generative and discriminative
  information retrieval models,'' in \emph{Proceedings of the 40th
  International ACM SIGIR conference on Research and Development in Information
  Retrieval}.\hskip 1em plus 0.5em minus 0.4em\relax ACM, 2017, pp. 515--524.

\bibitem{wang2020reinforced}
X.~Wang, Y.~Xu, X.~He, Y.~Cao, M.~Wang, and T.-S. Chua, ``Reinforced negative
  sampling over knowledge graph for recommendation,'' in \emph{Proceedings of
  The Web Conference 2020}, 2020, pp. 99--109.

\bibitem{park2019adversarial}
D.~H. Park and Y.~Chang, ``Adversarial sampling and training for
  semi-supervised information retrieval,'' in \emph{The World Wide Web
  Conference}.\hskip 1em plus 0.5em minus 0.4em\relax ACM, 2019, pp.
  1443--1453.

\bibitem{dong2020equivalence}
H.~Dong, J.~Chen, F.~Feng, X.~He, S.~Bi, Z.~Ding, and P.~Cui, ``On the
  equivalence of decoupled graph convolution network and label propagation,''
  \emph{arXiv preprint arXiv:2010.12408}, 2020.

\end{thebibliography}
%
%
%
%
%
%
%
%




\end{document}